\documentclass[11pt,reqno]{amsart} 
\usepackage{amsthm}
\usepackage{amsmath}
\usepackage{amsfonts}
\usepackage{amssymb}
\usepackage[dvips,all]{xy}
\usepackage{graphicx}
\usepackage{url}
\usepackage{enumerate}
\usepackage{lscape}
\usepackage[usenames,dvipsnames]{color}
\usepackage{multicol}
\usepackage{mathtools}
\usepackage{array}
\usepackage{blkarray}
\usepackage{subfig}

\usepackage{xcolor}

\newtheorem{theorem}{Theorem}[section]
\newtheorem{corollary}[theorem]{Corollary}
\newtheorem{proposition}[theorem]{Proposition}
\newtheorem{lemma}[theorem]{Lemma}
\newtheorem{example}[theorem]{Example}

\newtheorem{definition}[theorem]{Definition}

\newtheorem*{question*}{Question} 
\newtheorem*{problem*}{Problem} 
\theoremstyle{remark}
\newtheorem{remark}[theorem]{Remark}

\numberwithin{equation}{section}
\newcommand{\Rplus}{\mathbb{R}_{>0}}
\newcommand{\Rnn}{\mathbb{R}_{\geq 0}}

\newcommand{\Edg}{E}

\newcommand{\St}{\mathcal{S}}

\newcommand{\CRS}{chemical reaction system }

\newcommand{\invtPoly}{\mathcal{P}}

\newcommand{\R}{\mathbb{R}}

\definecolor{dgreen}{rgb}{.2,.6,.2}
\colorlet{darkgreen}{black!30!dgreen}

\definecolor{dblue}{rgb}{0.0,0.0,0.68}

\DeclareMathOperator{\im}{im}
\DeclareMathOperator{\sign}{sign}

\begin{document} 

\title{A global convergence result for processive multisite phosphorylation systems}
\author{Carsten Conradi and Anne Shiu}
{\address{	
	CC: Max-Planck-Institut Dynamik komplexer technischer Systeme, 
    Sandtorstr.\ 1, 39106 Magdeburg, Germany;
    AS: Dept.\ of Mathematics,
    University of Chicago, 5734 S.\ University Ave., Chicago IL 60637 USA, present address:
    Dept.\ of Mathematics, Mailstop 3368, Texas A\&M Univ., College Station, TX~77843--3368, USA
}}
\footnote{AS was supported by the NSF
  (DMS-1004380 and DMS-1312473).  CC was supported in part from BMBF
  grant Virtual Liver (FKZ 0315744) and the research focus dynamical
  systems of the state Saxony-Anhalt.}

\email{conradi@mpi-magdeburg.mpg.de,annejls@math.tamu.edu}

\date{\today}

\begin{abstract}
Multisite phosphorylation plays an important role in intracellular 
signaling.  There has been much recent work aimed at understanding the 
dynamics of such systems when the phosphorylation/dephosphorylation 
mechanism is distributive, that is, when the binding of a substrate 
and an enzyme molecule results in addition or removal of a single 
phosphate group and repeated binding therefore is required for 
multisite phosphorylation.  In particular, such systems admit 
bistability.  Here we analyze a different class of multisite systems, 
in which the binding of a substrate and an enzyme molecule results in 
addition or removal of phosphate groups at all phosphorylation sites. 
That is, we consider systems in which the mechanism is  processive, 
rather than  distributive. We show that in contrast with distributive 
systems, processive systems modeled with mass-action kinetics do not 
admit bistability and, moreover, exhibit rigid dynamics: each 
invariant set contains a unique equilibrium, which is a global 
attractor.  Additionally, we obtain a monomial parametrization of the 
steady states. 
Our proofs rely on a technique of Johnston for using ``translated''
networks to study systems with ``toric steady states'', recently given
sign conditions for injectivity of polynomial maps, and a result from
monotone systems theory due to Angeli and Sontag.
\vskip 0.1cm
\noindent \textbf{Keywords:} reaction network, mass-action kinetics,
multisite phosphorylation, global convergence, steady state, monomial
parametrization, monotone systems
\end{abstract}

\maketitle

\section{Introduction}\label{sec:intro}
A biological process of great importance, phosphorylation is the
enzyme-mediated addition of a phosphate group to a protein substrate,
which often modifies the function of the substrate.
Additionally, many such substrates have more than one site at which
phosphate groups can be attached.  Such multisite phosphorylation
systems may be {\em distributive} or {\em processive}. In distributive
systems, each 
  enzyme-substrate binding
results in 
one addition or removal of a phosphate group,
whereas in processive systems, when an enzyme catalyzes the addition
or removal of a phosphate group, then phosphate groups are added or
removed from all available sites before the enzyme and substrate
dissociate.
  The (fully) processive and distributive mechanisms can be viewed
  as the extremes of a whole spectrum of possible
  mechanisms \cite{sig-041}. Some proteins are phosphorylated at $N>1$
  sites with each enzyme-substrate binding, but not necessarily at all
  available sites. An example that is briefly discussed in
  \cite{sig-041} is the yeast transcription factor Pho4. With every
  enzyme-substrate binding it is on average phosphorylated at two
  sites (cf.~\cite{sig-041} and references therein). Other proteins,
  however, are phosphorylated at all available sites in a single
  encounter with the kinase (examples are the splicing factor ASF/SF2
  or the Crk-associated substrate (Cas), cf.\ \cite{PM,sig-041}). For
  more biological examples and discussion of the biological
  significance of multisite phosphorylation, we refer the reader to
  the work of Salazar and H{\"o}fer \cite{sig-041} and of
  Gunawardena~\cite{Guna_threshold,Guna}. An excellent 
  source for processive systems in particular is the review article of
  Patwardhan and Miller~\cite[\S 2--5]{PM}.

Ordinary differential equations (ODEs) frequently are used to
describe the dynamics of the chemical species involved in multisite
phosphorylation, e.g.\ protein substrate, (partially) phosphorylated
substrate, catalyzing enzymes, enzyme-substrate complexes,
and so on. A protein substrate can have many phosphorylation sites
(examples are discussed in \cite{TG2}), and the number of variables
and parameters increases with the number of phosphorylation
sites. Detailed models of multisite phosphorylation are therefore
large, while only a limited number of variables can be
measured. Thus, parameter values can only be specified within large
intervals, if at all (that is, parameter uncertainty is high). For
all these reasons, mathematical analysis of models of multisite
phosphorylation typically requires analysis of parametrized families 
of ODEs.

Much of the prior work on the mathematics of phosphorylation systems 
has focused on parametrized families of ODEs that describe
multisite phosphorylation under a sequential and {\em distributive}
mechanism; for instance, see Conradi {\em et al.}~\cite{MAPK},
Conradi and Mincheva~\cite{a6maya}, Feliu and
Wiuf~\cite{enzyme-sharing}, Hell and Rendall~\cite{bistable},
Holstein {\em et al.}~\cite{KathaMulti}, Flockerzi {\em et
  al.}~\cite{FHC}, 
Manrai and Gunawardena~\cite{ManraiGuna}, 
  Markevich {\em et al.}~\cite{Markevich04}, 
P\'erez Mill\'an {\em et al.}~\cite{TSS}, 
Thomson and Gunawardena~\cite{TG,TG2},
and 
Wang and Sontag \cite{WangSontag}. 
Here we concentrate instead on the multisite phosphorylation by a
kinase/phosphatase pair in a sequential and {\em processive}
mechanism, building on work of Gunawardena~\cite{Guna} and
Conradi~{\em et al.}~\cite{ConradiUsing}.
  While models of distributive phosphorylation can admit multiple
  steady states and multistability whenever there are at least two 
  phosphorylation sites~\cite{bistable,TG2,WangSontag}, it was shown in
  \cite{ConradiUsing} that models of processive phosphorylation at
  two phosphorylation sites cannot admit more than one steady state in
  each invariant set.  Whether or not this holds for models with an
  arbitrary number of phosphorylation sites has not been discussed
  previously. Also, it is known from \cite{TG} that there exists a
  rational
  parametrization of the set of all positive steady states for processive systems. However,
  no explicit parametrization has been given for an arbitrary number of
  phosphorylation sites.  
  
  The present article addresses both of the aforementioned topics: the number of steady states and a parametrization of steady states of processive phosphorylation systems.  
  We show that processive phosphorylation
  belongs to the class of chemical reaction systems with ``toric steady
  states'' (as does distributive phosphorylation~\cite{TSS} and other related networks~\cite{MPM_MAPK}) and present a
  particular (rational) parametrization of all positive steady
  states (Proposition~\ref{prop:parametrization}). By applying the Brouwer fixed-point theorem and recent
  results on injectivity of polynomial maps, 
we conclude
  that every
  invariant set contains a unique element of this parametrization and
  that this element is the only steady state within the invariant
  set (Theorem~\ref{thm:existence-uniqueness-via-signs}). Finally, in Theorem~\ref{thm:global_conv}, we prove -- for every
  invariant set -- global convergence to that unique steady
  state by applying a result of Angeli and Sontag from monotone
  systems theory~\cite{AS}.


The outline of our paper is as follows.  In Section~\ref{sec:dyn_sys},
we introduce the dynamical systems arising from chemical reaction
networks taken with mass-action kinetics.  The networks of interest in
this work, those arising from multisite phosphorylation by a
sequential and processive mechanism, are introduced in
Section~\ref{sec:processive}.  In Section~\ref{sec:augment}, we
describe a ``translated'' version of the network which will aid our
analysis.  In Section~\ref{sec:steady_state}, we prove the existence
and uniqueness of steady states of the processive multisite system
taken with mass-action kinetics and obtain a monomial parametrization
of the steady states. 
Global stability is proven in Section~\ref{sec:convergence}, 
and a discussion appears in Section~\ref{sec:discussion}.

\section{Dynamical systems arising from chemical reaction networks}
\label{sec:dyn_sys}
In this section we recall how a chemical reaction network gives rise to a dynamical system, beginning with an illustrative example.  
An example of a {\em chemical reaction}, as it usually appears in the
literature, is the following:
\begin{align} \label{reaction_ex}
 \begin{xy}<15mm,0cm>:
 (1,0) ="3A+C" *+!L{3A+C} *{};
 (0,0) ="A+B" *+!R{A+B} *{};
 (0.55,.05)="k" *+!D{\kappa} *{};
   {\ar "A+B"*{};"3A+C"*{}};
   \end{xy}
\end{align} 
In this reaction, one unit of chemical {\em species} $A$ and one of $B$ react 
to form three units of $A$ and one of $C$.  
The {\em educt} (or {\em reactant}) $A+B$ and the {\em product} $3A+C$ are called {\em complexes}. 
The concentrations of the three species, denoted by $x_{A},$ $x_{B}$, and
$x_{C}$, 
will change in time as the reaction occurs.  Under the assumption of {\em
mass-action kinetics}, species $A$ and $B$ react at a rate proportional to the
product of their concentrations, where the proportionality constant is the {\em reaction rate
constant} $\kappa$.  Noting that the reaction yields a net change of two units in
the amount of $A$, we obtain the first differential equation in the following
system:
\begin{align*}
\frac{d}{dt}x_{A}~&=~2\kappa x_{A}x_{B}~ \\
\frac{d}{dt}x_{B} ~&=~-\kappa x_{A}x_{B}~ \\
\frac{d}{dt} x_{C}~&=~\kappa x_{A}x_{B}~.
\end{align*}
The other two equations arise similarly.  A {\em chemical reaction network}
consists of finitely many reactions.  
The mass-action differential equations that a network defines are comprised of a 
sum of the monomial contribution from the reactant of each 
chemical reaction in the network; these differential equations 
will be defined by equations~(\ref{eq:ODE}--\ref{eq:R-for-mass-action}).


\subsection{Chemical reaction systems}
We now provide precise definitions.  A {\em chemical reaction network} 
consists of a finite set of species $\{A_1,A_2,\dots, A_s\}$, a finite set of 
complexes (finite nonnegative-integer combinations of the species), and 
a finite set of reactions (ordered pairs of the complexes).
A chemical reaction network is often depicted by  
a finite directed graph 
whose vertices are labeled by complexes and whose edges 
correspond to reactions.
Specifically, the digraph is denoted $G = (V,\Edg)$, with vertex set $V = \{1,2,\ldots,p\}$
and edge set $\,\Edg \subseteq \{(i,j) \in V \times V : \,i\not= j \}$.  
Throughout this paper, the integer unknowns~$p$, $s$, and $r$ denote the numbers of
complexes, species, and reactions, respectively.  
Writing the $i$-th complex as $y_{i1} A_1 + y_{i2} A_2 + \cdots + y_{is}A_s$ (where $y_{ij} \in \mathbb{Z}_{\geq 0}$ for $j=1,2,\dots,s$), 
we introduce the following monomial:
$$ x^{y_i} \,\,\, := \,\,\, x_1^{y_{i1}} x_2^{y_{i2}} \cdots  x_s^{y_{is}}~. $$
For example, the two complexes in~\eqref{reaction_ex} give rise to 
the monomials $x_{A}x_{B}$ and $x^3_A x_C$, which determine the vectors 
$y_1=(1,1,0)$ and $y_2=(3,0,1)$.  
These vectors define the rows of a $p \times s$-matrix of nonnegative integers,
which we denote by $Y=(y_{ij})$.
Next, the unknowns $x_1,x_2,\ldots,x_s$ represent the
concentrations of the $s$ species in the network,
and we regard them as functions $x_i(t)$ of time $t$.

A directed edge $(i,j) \in \Edg$ represents a reaction 
{$y_i \to
  y_j$} from the $i$-th chemical complex to the $j$-th chemical
complex, and the {\em reaction vector}
 $y_j-y_i$ encodes the
net change in each species that results when the reaction takes
place.  
Also, 
associated to each edge is 
a positive parameter $\kappa_{ij}$, the rate constant of the
reaction.  In this article, we will treat the rate constants $\kappa_{ij}$ as positive
unknowns in order to analyze the entire family of dynamical systems
that arise from a given network as the $\kappa_{ij}$'s vary.  A network
is said to be \emph{weakly reversible} if every connected component of
the network is strongly connected.

A pair of {\em reversible} reactions refers to a bidirected edge
$y_i \rightleftharpoons y_j$ in $E$.  
For each such pair $y_i \rightleftharpoons y_j$,
we designate 
a `forward' reaction $y_i \to y_j$ and a `backward' reaction $y_i
\leftarrow y_j$.
Letting $m$ denote the number of reactions, where we count
each pair of reversible reactions only once, 
the {\em stoichiometric matrix}  
$\Gamma$ is the $s \times m$ matrix whose $k$-th column 
is the reaction vector of the $k$-th reaction 
 (in the forward direction if the reaction is reversible),
i.e., it is the vector $y_j - y_i$ if $k$ indexes the (forward)
reaction $y_i \to y_j$.
The choice of kinetics is encoded by a locally Lipschitz function $R:\Rnn^s \to \R^m$ that encodes the reaction rates of the $m$ reactions as functions of the $s$ species concentrations (a pair of reversible reactions is counted only once -- in this case, $R_k$ is the forward rate minus the backward rate).  
The {\em reaction kinetics system} 
defined by a reaction network $G$ and reaction rate function $R$ is given by the following system of ODEs:
\begin{align} \label{eq:ODE}
\frac{dx}{dt} ~ = ~ \Gamma \,  R(x)~.
\end{align}
For {\em mass-action kinetics}, which is the setting of this paper, the coordinates of $R$ are:
\begin{equation}\label{eq:R-for-mass-action}
 R_k(x)=\left\lbrace
 \begin{array}{ll}
  \kappa_{ij} x^{y_i} & \textrm{ if $k$ indexes an irreversible reaction $y_i \to y_j$} \\
  \kappa_{ij} x^{y_i} - \kappa_{ji} x^{y_j}  & \textrm{ if $k$ indexes a reversible reaction $y_i \leftrightarrow y_j$} \\
 \end{array}\right. 
\end{equation}

A {\em chemical reaction system} refers to the 
dynamical system (\ref{eq:ODE}) arising from a specific chemical reaction
network $G$ and a choice of rate parameters $(\kappa^*_{ij}) \in
\mathbb{R}^{r}_{>0}$ (recall that $r$ denotes the number of
reactions) where the reaction rate function $R$ is that of mass-action
kinetics~\eqref{eq:R-for-mass-action}.


\begin{example} \label{ex:n=1}
  The following network (called the ``futile cycle'') describes 1-site
  phosphorylation: 
  \begin{equation}
    \label{eq:network_n=1}
    \xymatrix{
      S_0 + K  \ar @<.4ex> @{-^>} [r] ^{k_1} 
      &\ar @{-^>} [l] ^{k_2} S_0 K  \ar [r] 
      ^{k_{3}} 
      & S_1 + K\\ 
      S_1 + F \ar @<.4ex> @{-^>} [r] ^{\ell_3}
      &\ar @{-^>} [l] ^{\ell_{2}} S_{1} F \ar [r] 
      ^{\ell_{1}}
      &S_0 + F
    } 
  \end{equation}
  The key players in this network
  are a kinase ($K$), 
  a phosphatase ($F$), and a substrate ($S_0$).  
  The substrate $S_1$ is obtained from the unphosphorylated protein $S_0$ by 
  attaching a phosphate group to it via an enzymatic reaction catalyzed by $K$. 
  Conversely, a reaction catalyzed by $F$ removes the phosphate group from $S_1$ 
  to obtain $S_0$. The intermediate complexes~$S_0 K$ and~$S_1 F$ are the 
  bound enzyme-substrate complexes.
    Using the variables $x_1,x_2, \ldots, x_6$ to denote the species
    concentrations 
      $K, F,  S_0, S_1, S_0 K, S_1 F$, 
respectively, 
and letting $r_i$ denote the reaction vectors, the \CRS defined by the 1-site phosphorylation network~\eqref{eq:network_n=1} is given by the following ODEs:
    \begin{equation}
      \begin{split}
        \label{eq:odes_n1}
        \frac{dx}{dt} &= k_1\, x_1\, x_3\,
        \underbrace{
          \begin{pmatrix}
            -1 \\ \phantom{-}0 \\ -1 \\ \phantom{-}0 \\ \phantom{-}1\\
            \phantom{-}0 
          \end{pmatrix}
        }_{r_1}
        + k_2\, x_5\,
        \underbrace{
          \begin{pmatrix}
            \phantom{-}1 \\ \phantom{-}0 \\ \phantom{-}1 \\ \phantom{-}0
            \\ -1 \\\phantom{-}0  
          \end{pmatrix} 
        }_{r_2}
        + k_3\, x_5\,
        \underbrace{
          \begin{pmatrix}
            \phantom{-}1 \\ \phantom{-}0 \\ \phantom{-}0\\ \phantom{-}1\\
            -1\\ \phantom{-}0 
          \end{pmatrix}  
        }_{r_3}\\
        &\qquad
        +\ell_3\, x_2\, x_4\,
        \underbrace{ 
          \begin{pmatrix}
            \phantom{-}0 \\ -1 \\ \phantom{-}0 \\ -1 \\\phantom{-}0 \\
            \phantom{-}1 
          \end{pmatrix} 
        }_{r_4}
        + \ell_2\, x_6\,
        \underbrace{
          \begin{pmatrix}
            \phantom{-}0 \\ \phantom{-}1 \\ \phantom{-}0 \\ \phantom{-}1
            \\\phantom{-}0 \\ -1 
          \end{pmatrix} 
        }_{r_5}
        + \ell_1\, x_6\,
        \underbrace{
          \begin{pmatrix}
            \phantom{-}0 \\ \phantom{-}1 \\ \phantom{-}1 \\ \phantom{-}0
            \\\phantom{-}0 \\ -1
          \end{pmatrix} 
        }_{r_6}
      \end{split}
    \end{equation}
To recognize the above ODEs~\eqref{eq:odes_n1} in the general form~\eqref{eq:ODE}, we choose 
for the reversible reactions $S_0+K \leftrightharpoons S_0 K$ and
    $S_1+F\leftrightharpoons S_1 F$, the reactions $S_0+K \to S_0 K$ and
    $S_1+F\to S_1 F$ as forward reactions, and then obtain the following equivalent representation of the ODEs~(\ref{eq:odes_n1}):
    \begin{equation}
      \label{eq:odes_n=1_matrix_form}
      \frac{dx}{dt} = 
      \underbrace{
        \left[
          \begin{array}{*{4}{r}}
            -1 & 1 & 0 & 0 \\
            0 & 0 & -1 & 1\\
            -1 & 0 & 0 & 1 \\
            0 & 1 & -1 & 0 \\
            1 & -1& 0 & 0 \\
            0 & 0 & 1 & -1
          \end{array}
        \right]
      }_{=\Gamma}\,
      \underbrace{
        \begin{pmatrix}
          k_1\, x_1\,x_3 - k_2\, x_5 \\
          k_3\, x_5 \\
          \ell_3\, x_2\, x_4 - \ell_2\, x_6 \\
          \ell_1\, x_6
        \end{pmatrix}
      }_{=R(x)} .
    \end{equation}
    The column vectors of the stoichiometric matrix $\Gamma$ are $r_1$,
    $r_3$, $r_4$, and $r_6$. 
    We will study generalizations of the 
    chemical
    reaction system~\eqref{eq:odes_n=1_matrix_form} in this article. 
\end{example}

The {\em stoichiometric subspace} is the vector subspace of
$\mathbb{R}^s$ spanned by the reaction vectors
$y_j-y_i$ (where $(i,j)$ is an edge of $G$), and we will denote this
space by $\St$: 
\begin{equation} \label{eq:stoic_subs}
  \St~:=~ \mathbb{R} \{ y_j-y_i~|~ (i,j) \in \Edg \}~.
\end{equation}
  Note that in the setting of (\ref{eq:ODE}), one has $\St = \im(\Gamma)$.
In the earlier example reaction shown in~\eqref{reaction_ex}, we have $y_2-y_1 =(2,-1,1)$, which
means 
that with each occurrence of the reaction, two units of $A$ and one of $C$ are
produced, while one unit of $B$ is consumed.  This vector $(2,-1,1)$ spans the
stoichiometric subspace $\St$ for the network~\eqref{reaction_ex}. 
Note that the  vector $\frac{d x}{dt}$ in  (\ref{eq:ODE}) lies in
$\St$ for all time $t$.   
In fact, a trajectory $x(t)$ beginning at a positive vector $x(0)=x^0 \in
\R^s_{>0}$ remains in the {\em stoichiometric compatibility class} (also called an 
``invariant polyhedron''), which we
denote by
\begin{align}\label{eqn:invtPoly}
\invtPoly~:=~(x^0+\St) \cap \mathbb{R}^s_{\geq 0}~, 
\end{align}
for all positive time.  In other words, this set is forward-invariant with
respect to the dynamics~(\ref{eq:ODE}).    
A {\em steady state} of a reaction kinetics system~\eqref{eq:ODE} is a nonnegative concentration vector $x^* \in \Rnn^s$ at which the ODEs~\eqref{eq:ODE}  vanish: $\Gamma R(x^*) = 0$.  
We distinguish between {\em positive steady states} $x^* \in \mathbb{R}^s_{> 0}$ and {\em boundary steady
states} $x^* \in  \left( \mathbb{R}^s_{\geq 0} \setminus \mathbb{R}^s_{> 0} \right)$.  
A system is {\em multistationary} (or {\em admits multiple steady states}) if there exists a stoichiometric
compatibility class $\invtPoly$ with two or more positive steady states. 
In the setting of mass-action kinetics, a network may admit multistationarity for all, some, or no choices of
positive rate constants $\kappa_{ij}$. 

\subsection{Alternate description of chemical reaction systems}
We now give another characterization of the ODEs 
arising from 
mass-action kinetics
that will be useful for obtaining parametrizations of steady states.
  First we introduce the following monomial mapping
  defined by the row vectors of a nonnegative matrix $B\in\R^{p\times
    s}$:
  \begin{equation}
    \begin{split}
      \label{eq:Psi_def}
      \Psi^{(B)} &: \Rnn^s \to \Rnn^p, \\
      \Psi^{(B)}(x) &= ~ \bigl( x^{b_1}, ~ x^{b_2} , ~ \ldots ~ ,  ~ x^{b_p} \bigr)^t~
    \end{split}
  \end{equation}
  Second, recall that $Y$ is the $p \times s$-matrix with rows given by
  the $y_i$'s;
  following (\ref{eq:Psi_def}) these rows define the following
  monomial mapping:
$$ \Psi^{(Y)}(x) ~ = ~ \bigl( x^{y_1}, ~ x^{y_2} , ~ \ldots ~ ,  ~ x^{y_p} \bigr)^t~.$$
%
{Third,} let $A_\kappa$ denote the negative of the {\em
  Laplacian} of the chemical reaction network $G$. In other words,
$A_\kappa$ is the $p \times p$-matrix whose off-diagonal entries 
are the $\kappa_{ij}$ and whose row sums are zero.  
An equivalent characterization of the 
\CRS (\ref{eq:ODE}--\ref{eq:R-for-mass-action}) is
\begin{equation}
  \label{CRN2}
  \frac{dx}{dt} ~=~	Y^t ~  A_\kappa^t ~  \Psi^{(Y)}(x)~.
\end{equation}
That is, after fixing orderings of the species, complexes, and
reactions; the products $\Gamma\, R(x)$ and $Y^t \, A_\kappa^t
\, \Psi^{(Y)}(x)$ evaluate to the same polynomial system:
$    \Gamma\, R(x) = Y^t \, A_\kappa^t\, \Psi^{(Y)}(x).
$

\begin{example} 
  \label{ex:n=1-continued}
  We revisit the 1-site phosphorylation network~\eqref{eq:network_n=1}. 
  Using the ordering of the species given earlier 
(namely, $K, F, S_0, S_1,  S_0 K,  S_1 F$)
and the
  following ordering of the complexes:
  $S_0+K$, $S_0 K$, $S_1+K$, $S_1+F$, $S_1F$, $S_0+F$,
  the alternate description~\eqref{CRN2} of the chemical reaction
  system~\eqref{eq:odes_n=1_matrix_form} arises as the product of the
  following: 
  \begin{equation*}
    \Psi^{(Y)}(x) ~=~ \left(  x_1 x_3,~ x_5, ~x_1 x_4, ~x_2 x_4, ~x_6,
      ~x_2 x_3 \right)^t\ , 
  \end{equation*}
  \begin{equation*}
    Y^t ~=~ \left[  
      \begin{array}{llllllllll}
        1 & 0 & 1 & 0 & 0 & 0 \\ 
        0 & 0 & 0 & 1 & 0 & 1 \\
        1 & 0 & 0 & 0 & 0 & 1 \\
        0 & 0 & 1 & 1 & 0 & 0 \\
        0 & 1 & 0 & 0 & 0 & 0 \\
        0 & 0 & 0 & 0 & 1 & 0 
      \end{array}
    \right], {\rm ~ and}
  \end{equation*}
  \begin{equation*}
    A^t_{\kappa} ~:=~ 
    \left[  
      \begin{array}{*{3}{c}|*{3}{c}}
        -k_1 & k_2 & 0 & & &  \\
        k_1 & -(k_2+k_3) & 0 & & 0 & \\
        0 & k_3 & 0 & & & \\ \hline
        & & & -\ell_3 & \ell_2 & 0 \\
        & 0 & & \ell_3 & -(\ell_2+\ell_1) & 0 \\
        & & & 0 & \ell_1 & 0
      \end{array}
    \right].
  \end{equation*}
\end{example}

\subsection{Translated chemical reaction networks}
\label{sec:translatedRN}

Recall that the reaction vector encodes the net change in each species
when a given reaction takes place. If the same amount of some chemical
species is added to both the product and the educt complex of this
reaction, then the reaction vector is unchanged. For example, the
reactions $A+B \to 3\, A +C$ and $2A+B \to 4\, A +C $ both have the
reaction vector $(2,-1,1)^t$. Thus, if both reactions are assigned the
reaction rate function  $v=\kappa\, x_A\, x_B$, then both reactions
define the same dynamical system~\eqref{eq:ODE}.
To exploit this observation, 
Johnston introduced the notion of a \emph{translated chemical reaction
  network} in \cite{translated}: a translated chemical reaction
network is a reaction network obtained by adding to the product and
educt of each reaction the same amounts of certain species. 

Obviously, a given reaction network generates infinitely many
translated networks.  Here, we are interested only in those for
which 
the original network  
and its translation 
define the same dynamical system.
Translated networks for which this is possible include those
that are weakly reversible and for which there exists a
reaction-preserving bijection between the educt complexes in the
original network and those of the translated network
\cite[Lemma~4.1]{translated}. 
  In \cite{translated} such a weakly reversible translation is
  called proper.


  Next we consider a reaction
  network with matrices $A_\kappa$ and $Y$, and let $\tilde A_\kappa$ and
  $\tilde Y$ be the corresponding matrices defined by its proper,
  weakly reversible translated network.
  By \cite[Lemma~4.1]{translated}, there exists a matrix ${\mathcal Y}$
  such that the \CRS defined by the translation (taken with the
  monomial function $\Psi^{({\mathcal Y})}$)
  is identical to the \CRS defined by
  the original network (taken with 
  $\Psi^{(Y)}$), where the rate constants are taken to be the same:
\begin{equation}
  \label{eq:stron_proper}
  \tilde Y^t\, \tilde A^t_\kappa\, \Psi^{({\mathcal{Y}})}(x) = Y^t\,
  A^t_\kappa\, \Psi^{(Y)}(x) ~.
\end{equation}
For completeness, this entails 
$ 
 \tilde Y^t\, \tilde A^t_\kappa\, \Psi^{(\mathcal{Y})}(x) = \Gamma \, R(x), 
$
where $\Gamma$ is the stoichiometric matrix and $R$ the mass-action
rate function of the original network. 
In Section~\ref{sec:augment}, we will establish 
proper, weakly reversible
translations for the generalizations of
network~(\ref{eq:network_n=1}) described in
Section~\ref{sec:processive}.  And in Section~\ref{sec:steady_state}
we will obtain parametrizations of steady states based on these
translations. 
\section{Sequential and processive phosphorylation/dephosphorylation at $n$ sites} \label{sec:processive}
This section introduces a generalization of the 1-site phosphorylation network~\eqref{eq:network_n=1} to an $n$-site network.  In nature, an enzyme may facilitate the (de)phosphorylation of a substrate at $n$ sites by a processive or distributive mechanism.  Our work focuses on the processive mechanism; a comparison with the distributive mechanism appears in Subsection~\ref{sec:comparison_dist}.  

\subsection{Description of the processive $n$-site network}
Here is the reaction network that describes the sequential\footnote{In {\em sequential} (de)phosphorylation, phosphate groups are added or removed in a prescribed order.} and processive phosphorylation/dephosphorylation of a substrate at $n$ sites, which we call the {\em processive $n$-site network} in this paper: 
\begin{equation}
  \label{eq:network_processive}
  \begin{split}
  \xymatrix{
    S_0 + K  \ar @<.4ex> @{-^>} [r] ^-{k_1}
    &\ar @{-^>} [l] ^-{k_{2}} S_0 K \ar @<.4ex> @{-^>} [r] ^-{k_3}
    &\ar @{-^>} [l] ^-{k_{4}} S_1 K \ar @<.4ex> @{-^>} [r] ^-{k_{5}}
    &\ar @{-^>} [l] ^-{k_{6}} \hdots \ar @<.4ex> @{-^>} [r] ^-{k_{2n-1}}
    &\ar @{-^>} [l] ^-{k_{2n}} S_{n-1} K \ar [r] 
    ^-{k_{2n+1}}
    & S_n + K\\
    S_n + F \ar @<.4ex> @{-^>} [r] ^-{\ell_{2n+1}}
    &\ar @{-^>} [l] ^-{\ell_{2n}} S_{n} F \ar @<.4ex> @{-^>} [r]
    ^-{\ell_{2n-1}} 
    &\ar @{-^>} [l] ^-{\ell_{2n-2}} \hdots \ar @<.4ex> @{-^>} [r]
    ^-{\ell_5}
    &\ar @{-^>} [l] ^-{\ell_4} S_2 F \ar @<.4ex> @{-^>} [r]
    ^-{\ell_3}
    &\ar @{-^>} [l] ^-{\ell_2} S_1 F \ar [r]
    ^-{\ell_1}
    &S_0 + F
  }
  \end{split}
\end{equation}

We see that the substrate undergoes $n>1$
phosphorylations after binding to the kinase and forming the
enzyme-substrate complex; thus, only the fully
phosphorylated substrate is released and hence only two phosphoforms
have to be considered: the unphosphorylated substrate $S_0$ and fully phosphorylated
substrate $S_n$ (see,
for example, \cite{Markevich04,sig-041}). 
Processive dephosphorylation proceeds similarly.
%
The enzyme-substrate complex formed by the kinase (or phosphatase, respectively)
 and the substrate with $i$ phosphate groups attached is denoted by $S_i K$ (or $S_i F$, respectively).
%

\begin{table}
  \centering
  \begin{tabular}{*{11}{>{$}c<{$}}}\hline
    x_1 & x_2 & x_3 & x_4 & x_5 & x_6 & x_7 & x_8 & \cdots & x_{2n+3}
    & x_{2n+4} \\ \hline
    K & F & S_0 & S_n & S_0 K & S_1 F & S_1 K & S_2 F &
    \cdots & S_{n-1} K & S_n F\\ \hline
  \end{tabular}
  \caption{Assignment of variables and species of the processive $n$-site network~\eqref{eq:network_processive}}
  \label{tab:variables}
\end{table}

Letting $e_i \in \mathbb{R}^{2n+4}$ denote the $i$-th standard basis vector, the 
$(2n+4) \times (2n+2)$ stoichiometric matrix for the $n$-site processive network~\eqref{eq:network_processive}
is the following, where the rows are indexed by the $2n+4$ species in the order presented in Table~\ref{tab:variables} and the 
columns correspond to the (forward) reactions in the order given by 
$(k_1, k_3, \dots, k_{2n+1}, ~ l_{2n+1}, l_{2n-1}, \dots, l_1)$:
\begin{equation}
  \begin{split}
    \label{eq:def_Gamma_proc}
    \Gamma &= \bigl[
    e_5 - (e_1+e_3) 
	\quad \lvert \quad  
     \ldots,\, e_{2i+5} - e_{2i+3},\, \ldots 
	\quad \lvert \quad
    e_4+e_1 - e_{2n+3},  \\ 
    & \quad \quad \qquad e_2+{ e_{3}} - e_6
    	\quad \lvert \quad  
     \ldots ,~
    e_{{2i+4}}-e_{{2i+6}},\, \ldots
    	\quad \lvert \quad  
    e_{2n+4}-(e_2+e_4) 
    \bigr]
  \end{split}
\end{equation}
where $i=1$, \ldots, $n-1$. The reaction rate function arising from mass-action kinetics~\eqref{eq:R-for-mass-action} is:
\begin{equation}
    \label{eq:def_R_proc}
    R(x) = 	
    \left[  
    \begin{array}{c}
    k_1 x_1 x_3 - k_2 x_5 \\
    \hline
	k_3 x_5 - k_4 x_7 \\ 	
	k_5 x_7 - k_6 x_9 \\ 	
	\vdots \\
	k_{2n-1} x_{2n+1} - k_{2n} x_{2n+3} \\ 	
	\hline
    k_{2n+1} x_{2n+3}\\
	\hline
	\ell_{2n+1} x_2 x_4 - \ell_{2n} x_{2n+4} \\
	\hline
	\ell_{2n-1} x_{2n+4} - \ell_{2n-2} x_{2n+2} \\
	\ell_{2n-3} x_{2n+2} - \ell_{2n-4} x_{2n} \\
	\vdots \\
	\ell_{3} x_{8} - \ell_{2} x_{6} \\	
	\hline
	\ell_{1} x_6
    \end{array}
	\right]~.
\end{equation}
For $n=1$, the matrices~(\ref{eq:def_Gamma_proc}--\ref{eq:def_R_proc})
appeared earlier in~\eqref{eq:odes_n=1_matrix_form}.

\begin{remark} \label{rmk:proc_net_vs_others}
  The processive multisite network~\eqref{eq:network_processive} 
  is consistent with the one presented in~\cite{ConradiUsing}, but 
  differs somewhat from the ones given in 
  \cite[Figure~1]{PM} and~\cite[Equation~(7)]{Guna}.  
\end{remark}

  Next we consider the rank of the matrix $\Gamma \in \mathbb{R}^{(2n+4)\times (2n+2)}$:
  \begin{lemma}\label{lem:rank-gamma}
    The matrix $\Gamma$ from (\ref{eq:def_Gamma_proc}) has rank 2n+1.
  \end{lemma}
  \begin{proof}
    It is easy to see that the row sums of $\Gamma$ are zero (so the
    rank is at most $2n+1$) because each species appears with
    stoichiometric coefficient one in the educt (reactant) of exactly
    one reaction (in the forward direction) and similarly in the
    product of exactly one reaction. Also, after reordering the rows
    so that the first $2n+1$ rows are indexed by the species as
    follows:  
    $(S_0, S_0 K, S_1 K, \dots, S_{n-1} K, S_n, S_n F, S_{n-1} F,
    \dots, S_2 F)$, the upper $(2n+1) \times (2n+1)$-submatrix has
    full rank. Indeed, this submatrix is lower-triangular with $-1$'s
    along the diagonal; this holds because when considering only the
    first $2n+1$ species, reaction 1 involves only $S_0$ as educt
    (reactant) and $S_0K$ as product (corresponding to rows 1 and 2,
    respectively), reaction 2 involves only $S_0K$ and $S_1K$ (rows 2
    and 3), and so on, with reaction $2n$ involving only rows $2n$ and
    $2n+1$ and reaction $2n+1$ involving only row $2n+1$.
  \end{proof}

\begin{remark}[Conservation relations]\label{rmk:cons_laws} \mbox{} 
  By Lemma~\ref{lem:rank-gamma}, 
    $\ker(\Gamma^t) = \mathcal{S}^{\perp}$ is three-dimensional.
    A particular basis is formed by the rows of the following matrix:
    \begin{align} \label{eq:matrix_A}
      \mathcal{A} = \left[
        \begin{array}{rr|rr|rrrrrrr}
          1 & 0 & 0 & 0 & 1 & 0 & 1 & 0 & \cdots & 1 & 0 \\
          0 & 1 & 0 & 0 & 0 & 1 & 0 & 1 & \cdots & 0 & 1 \\
          0 & 0 & 1 & 1 & 1 & 1 & 1 & 1 & \cdots & 1 & 1 \\
        \end{array}
      \right].
    \end{align}
   This basis has the following interpretation: the total amounts
    of free and bound enzyme or substrate remain constant as the dynamical system~\eqref{eq:ODE} progresses. In other words, 
the rows of $\mathcal{A}$ correspond to the following conserved (positive) quantities    
(recall the species ordering from Table~\ref{tab:variables}):
      \begin{align*}
        \label{eqn:conservation}
        K_{\mbox{tot}} &= x_1 + (x_5+x_7+ \cdots + x_{2n+3}), \\
        F_{\mbox{tot}} &= x_2 + (x_6+ x_8 + \cdots + x_{2n+4})   ~, \\
        S_{\mbox{tot}} &= x_3+x_4+\cdots +x_{2n+4}~.
      \end{align*}
\end{remark}

From the conservation relations, we establish that no boundary steady states exist, 
by a straightforward generalization of the analysis due
  to Angeli, De~Leenheer, and Sontag in~\cite[\S~6,  Ex.\ 1--2]{ADS07}.
  \begin{lemma}\label{lem:nobss}
  Let $x^* \in \R_{\ge 0}^{2n+4} - \R_{>0}^{2n+4}$ be a boundary steady
    state. Set $\Lambda:=\{ i \in \{1, \dots, 2n+4\} : x^*_i=0\}$. Then,  
    $\Lambda$ contains the support of at least one of the vectors
    defining the conservation relations~\eqref{eq:matrix_A}.
    Thus, there are no boundary steady states in any stoichiometric
    compatibility class. 
  \end{lemma}

\begin{remark}[Existence of steady states via the Brouwer fixed-point theorem] \label{rmk:brouwer}
The aim of this paper is to analyze the chemical reaction systems arising from the $n$-site phosphorylation network (for all $n$ and all choices of rate constants), that is, the dynamical system $\frac{dx}{dt} = \Gamma \, R(x)$, where $\Gamma$ and $R(x)$ are given in~(\ref{eq:def_Gamma_proc}--\ref{eq:def_R_proc}).  We will show that the steady states admit a monomial parametrization, each stoichiometric compatibility class has a unique steady state, and this steady state is a global attractor.  As a first step, the existence of at least one steady state in each compatibility class is guaranteed by the Brouwer fixed-point theorem (for details, see \cite[Remark~3.9]{TSS}); indeed, the compatibility classes are compact because of the conservation laws (Remark~\ref{rmk:cons_laws}) and there are no boundary steady states 
(Lemma~\ref{lem:nobss}). 
Therefore, to show that a unique steady state exists in each compatibility class, it suffices to preclude multistationarity.  This will be accomplished in Section~\ref{sec:steady_state}.
\end{remark}


\subsection{Comparison with distributive multisite systems} \label{sec:comparison_dist}
Here we describe, for comparison, the {\em distributive} multisite phosphorylation networks and what is known about their dynamics. Phosphorylation/dephosphorylation is distributive when the binding of a substrate and an enzyme results in at most one addition or removal of a phosphate group.  
The {\em distributive $n$-site network} describes the sequential and distributive phosphorylation/dephosphorylation of a substrate at $n$ sites:
\begin{equation}
  \label{eq:network_distributive}
  \begin{split}
  \xymatrix@R=0.5em@C=1.25em{
      S_0 + K  \ar @<.4ex> @{-^>} [r]
    &\ar @{-^>} [l] S_0 K \ar [r]
    & S_1 + K \ar @<.4ex> @{-^>} [r]
    &\ar @{-^>} [l] S_1 K \ar [r]
    &\hdots \ar @{->} [r]
    & S_{n-1} + K \ar @<.4ex> @{-^>} [r]
    &\ar @{-^>} [l] S_{n-1} K \ar [r]
    & S_n + K\\
      S_n + F  \ar @<.4ex> @{-^>} [r]
    &\ar @{-^>} [l] S_n F \ar [r]
    &\hdots \ar @{->} [r]
    & S_2 + F \ar @<.4ex> @{-^>} [r]
    &\ar @{-^>} [l] S_2 F \ar [r]
    & S_1 + F \ar @<.4ex> @{-^>} [r]
    &\ar @{-^>} [l] S_1 F \ar [r]
    & S_0 + F\\
  }
  \end{split}
\end{equation}

  For any $n\geq 2$, there exist rate constants such that the \CRS arising
from the distributive $n$-site network~\eqref{eq:network_distributive}
admits multiple steady states~\cite{FHC,bistable,TG2,WangSontag}. These rate
constants arise from the solutions of the linear inequality systems
described in \cite{KathaMulti}.
One goal of this work is to highlight the differences between
distributive systems and processive systems. In particular, as we will
see, processive systems are not multistationary: their steady states
are unique and global attractors (Theorem~\ref{thm:global_conv}).
Indeed, this confirms mathematically the observation in~\cite[\S
5]{PM} that distributive phosphorylation can be switch-like, while
processive phosphorylation is not.

Both distributive and processive systems have {\em toric steady
  states}: the set of steady states is cut out by binomials, which
then gives rise to a monomial parametrization of the steady states.
This was shown for distributive systems by P\'erez Mill\'an {\em et
  al.}~\cite[\S 4]{TSS}.  For processive systems, this will be
accomplished in Section~\ref{sec:steady_state}.

\section{Translated version of the processive network} \label{sec:augment}
Here we present a translated version of the processive
$n$-site network which will aid in our analysis of 
the steady states of the original
network (cf.\ Section~\ref{sec:translatedRN}).  
This network is obtained from the processive $n$-site network~\eqref{eq:network_processive} by adding $F$ to every complex of the first connected component and adding $K$ to every complex of
the second connected component:
\begin{equation}
  \label{eq:augmented_network_proc}
  \xymatrix@C=3ex{
    S_0 + K  +F \ar @<.4ex> @{-^>} [r] ^-{k_1}
    &\ar @{-^>} [l] ^-{k_{2}} S_0 K +F \ar @<.4ex> @{-^>} [r] ^{k_3}
    &\ar @{-^>} [l] ^{k_{4}} S_1 K +F \ar @<.4ex> @{-^>} [r] ^-{k_{5}}
    &\ar @{-^>} [l] ^-{k_{6}} \phantom{K}\cdots \phantom{K} \ar @<.4ex>
    @{-^>} [r] ^-{k_{2n-3}}
    &\ar @{-^>} [l] ^-{k_{2n-2}} S_{n-2} K +F \ar @<.4ex> @{-^>} [r]
     ^{k_{2n-1}}
    &\ar @{-^>} [l] ^{k_{2n}} S_{n-1} K +F \ar [d] 
    ^{k_{2n+1}} \\
    \ar [u]  ^{\ell_1} S_1 F + K \ar @<.4ex> @{-^>} [r] ^{\ell_2}
    & \ar @{-^>} [l] ^{\ell_3} S_2 F + K \ar @<.4ex> @{-^>} [r]
    ^-{\ell_4}
    &\ar @{-^>} [l] ^-{\ell_5} \phantom{K}\cdots\phantom{K}
    \ar @<.4ex> @{-^>} [r] ^-{\ell_{2n-4}}
    &\ar @{-^>} [l] ^-{\ell_{2n-3}} S_{n-1} F + K  \ar @<.4ex> @{-^>}
    [r] ^-{\ell_{2n-2}}
    &\ar @{-^>} [l] ^-{\ell_{2n-1}}  S_{n} F+K \ar @<.4ex> @{-^>} [r]
    ^-{\ell_{2n}} 
    &\ar @{-^>} [l] ^-{\ell_{2n+1}} S_n + K + F
  }
\end{equation}

Consisting of a single strongly connected component, our translated
network~\eqref{eq:augmented_network_proc} is therefore weakly reversible.  
Our subsequent arguments generalize the analysis of the 1-site network
by Johnston~\cite[Example I]{translated} and fits in the setting of
Theorem 4.1 in that work.

\begin{table}
  \centering
  \begin{tabular}{*{4}{>{$}c<{$}}}\hline
    \text{Complex} & 	\text{Corresponding vector } \tilde y_i
    & \text{Educt complex in
      (\ref{eq:network_processive})} &  \text{Corresponding vector } y_i \\ \hline \hline
    S_0 + K + F 	&     \tilde y_1 = e_1 + e_2 + e_3	&
    S_0+K  & y_1 = e_1 +e_3 \\ \hline
    \vdots & \vdots &  \vdots & \vdots \\
    S_i K +F &     \tilde y_{i+2} = e_2 + e_{2i+5}  & S_i K & y_{i+2}
    = e_{2i+5} \\ 
    \vdots & \vdots & \vdots & \vdots  \\ \hline
    S_n + K + F & \tilde y_{n+2} = e_1 + e_2 + e_4 & S_n + F &
    y_{n+2} = e_2 +e_4 \\ \hline
    \vdots & \vdots & \vdots & \vdots  \\
    S_{n-i} F + K & \tilde y_{n+i+3} = e_1 + {e_{2n+4-2i}} &
    S_{n-i} F & y_{n+i+3} = e_{2n+4-2i}  \\ 
    \vdots & \vdots & \vdots & \vdots \\ \hline	
  \end{tabular}
  \caption{Column 1: the complexes of the {translated}
    network~(\ref{eq:augmented_network_proc}); column 2: the  corresponding
    vectors $\tilde y_i$ (via the species ordering in
    Table~\ref{tab:variables}); column 3: the unique corresponding educt complexes of the (original) processive $n$-site network (\ref{eq:network_processive}); column 4: the corresponding
    vectors $y_i$. 
    The index $i$ runs over $0\leq i\leq n-1$.
  }
  \label{tab:complexes_aug_proc}
\end{table}
Following the ideas introduced in Section~\ref{sec:translatedRN}, we
establish in Table~\ref{tab:complexes_aug_proc} (columns 1 and 3) a
reaction-preserving bijection between educt complexes of the original
processive network~(\ref{eq:network_processive}) and those of its
translation~(\ref{eq:augmented_network_proc}). Hence, as explained in
Section~\ref{sec:translatedRN}, the
translation~(\ref{eq:augmented_network_proc}) is weakly reversible and
proper. Columns 2 and 4 of Table~\ref{tab:complexes_aug_proc} give the
vectors $\tilde y_i$ of the translation together with the
corresponding vectors $y_i$ of the original network. These vectors
define matrices $\tilde Y$ and  $\mathcal{Y}$:
\begin{align} \label{eq:two-matrices-Y}
    \tilde Y = \left[
      \begin{array}{c}
        \tilde y_1 \\ \vdots \\ \tilde y_{2n+2}
      \end{array}
    \right]\; \text{and } \quad
    \mathcal{Y} = \left[
      \begin{array}{c}
        y_1 \\ \vdots \\ y_{2n+2}
      \end{array}
    \right].
\end{align}
The matrix $\mathcal{Y}$ defines the monomial vector 
  \begin{equation}
    \label{eq:psi_processive}
    \Psi^{(\mathcal{Y})}(x) = \left( x_1 x_3  \,| \,  x_5,\, x_7,\, 
      \ldots, \, x_{2n+3} \, | \,
      x_2 x_4 \, | \, 
      x_{2n+4},\, x_{2n+2}, \,  \ldots, 
      \, x_6\right)^t ~.
  \end{equation}
Also, the matrix $\tilde A_\kappa^t \in \R^{(2n+2)\times (2n+2)}$ for the
translated network~(\ref{eq:augmented_network_proc}) is:
\begin{equation}
  \label{eq:def_Ak_aug}
  \arraycolsep=3pt 
  \medmuskip = 1mu 
  \tilde A_\kappa^t = 
  \scalebox{.55}{
    \mbox{
      \ensuremath{
        \begin{blockarray}{cccccccccccccc}
          & \scriptstyle 1 &\scriptstyle 2 &\scriptstyle 3 &\scriptstyle 4
          &\scriptstyle \hdots &\scriptstyle n+1 &\scriptstyle n+2
          &\scriptstyle n+3 &\scriptstyle n+4 &\scriptstyle n+5
          &\scriptstyle \hdots &\scriptstyle 2n+1 &\scriptstyle 2n+2 \\
          \begin{block}{>{\scriptstyle}c(ccccccccccccc)}
            1 & -k_1 & k_2 & & & & & & & & & & & l_1 \\
            2 & k_1 & -(k_2+k_3) & k_4 & & & & & & & & & & & \\
            3 & & k_3 & -(k_4+k_5) & k_6 & & & & & & & & &  \\
            4 & & & k_5 & -(k_6+k_7) & & & & & & & & & \\
            5 & & & & k_7 & & & & & & & & \\
            \vdots & & & & & \ddots & & & & & & & & \\
            n & & & & & & k_{2n} & & & & & & \\
            n+1 & & & & & & -(k_{2n} + k_{2n+1}) & & & & & & & \\
            n+2 & & & & & & k_{2n+1} & -\ell_{2n+1} & \ell_{2n} & & & & & \\
            n+3 & & & & & & & \ell_{2n+1} & -(\ell_{2n-1} + \ell_{2n}) & \ell_{2n-2} & & & &
            \\
            n+4 & & & & & & & & \ell_{2n-1} & -(\ell_{2n-3} + \ell_{2n-2}) & \ell_{2n-4} &
            & & \\
            n+5 & & & & & & & & &  \ell_{2n-3} & -(\ell_{2n-5} + \ell_{2n-4}) & & &
            \\
            n+6 & & & & & & & & &  & \ell_{2n-5} & & & \\
            \vdots & & & & & & & & & & & \ddots & & \\
            2n     & & & & & & & & & & & & \ell_{4} & \\
            2n+1   & & & & & & & & & & & & -(\ell_{3} + \ell_{4})
            & \ell_{2} \\
            2n+2 & & & & & & &  & & & & & \ell_{3} &
            -(\ell_{1} + \ell_{2}) \\
          \end{block}
        \end{blockarray}
      }
    }
  }
\end{equation}
As explained earlier, it follows that the \CRS defined by network~(\ref{eq:network_processive}) and the
  generalized \CRS defined by the
  translation~(\ref{eq:augmented_network_proc}) via the matrix
  $\mathcal{Y}$ 
  are identical~ \cite[Lemma~4.1]{translated}. That is, either system is defined by the following system of ODEs:
  \begin{equation}
    \label{eq:ss_via_translation}
    \frac{dx}{dt} 
    ~= ~
    \tilde Y^t\, \tilde A^t_\kappa\, \Psi^{(\mathcal{Y})}(x)
    ~=~ \Gamma \, R(x), 
  \end{equation}
  where 
  the matrix $\tilde Y$ is given in~\eqref{eq:two-matrices-Y} (via
  Table~\ref{tab:complexes_aug_proc}),  
 $\Psi^{(\mathcal{Y})}$ and  $\tilde A_\kappa$
are given in (\ref{eq:psi_processive}--\ref{eq:def_Ak_aug}), and
  the matrix $\Gamma$ and the function $R(x)$ arise from from the original network via mass-action kinetics and are given in
  (\ref{eq:def_Gamma_proc}--\ref{eq:def_R_proc}), respectively. 

We now analyze the matrix $\tilde Y \in\R^{(2n+2)\times (2n+4)}$ for
the translated network~\eqref{eq:augmented_network_proc}:
%
\begin{lemma}
  \label{lem:rank_Y_aug}
  The matrix $\tilde Y \in \R^{(2n+2) \times (2n+4)}$ for the
  {translated} network~\eqref{eq:augmented_network_proc}
  has
  full rank $2n+2$ and hence $\ker(\tilde Y^t)=0$.
\end{lemma}
\begin{proof}
  By Table~\ref{tab:complexes_aug_proc}, 
  \begin{displaymath}
    \tilde Y=
    \left[
      \begin{array}{c}
        e_1 + e_2 + e_3 \\ \hline
        e_2+e_5 \\ e_2 + e_7 \\
        \vdots \\ e_2+e_{2n+3} \\ \hline
        e_1+e_2+e_4 \\ \hline
        e_1+e_{2n+4} \\
        e_1+e_{2n+2} \\ \vdots \\ e_1 +e_6
      \end{array}
    \right]~.
  \end{displaymath}
  As $\tilde Y$ is a $(2n+2)\times(2n+4)$--matrix, it suffices to
  find $2n+2$ linearly independent columns.  Indeed, the submatrix
  $\hat Y$ consisting of the $2n+2$ columns $3$, 4, \ldots, $2n+4$ is
  a permutation matrix, so $\det(\hat Y) =\pm 1$. 
\end{proof}

\section{Existence and uniqueness of steady states}\label{sec:steady_state}

As mentioned earlier, the fully {\em distributive} $n$-site system
admits multiple steady states for all $n \geq 2$~\cite{TG2,
  WangSontag}.  
In this section, we show that the fully processive $n$-site systems preclude multiple steady states.
By equation (\ref{eq:ss_via_translation}), steady states $x \in
\mathbb{R}^{2n+4}_+$ of the processive system are characterized by the
following equivalent conditions:
\begin{equation}
  \label{eq:ss_kerAk}
  \Gamma \, R(x) = 0 
  \quad \Leftrightarrow \quad
  \tilde Y^t \tilde A^t_\kappa\,
  \Psi^{(\mathcal{Y})}(x) = 0 
  \quad \Leftrightarrow \quad 
  \tilde A^t_\kappa\, \Psi^{(\mathcal{Y})}(x) = 0~,
\end{equation}
where the rightmost equivalence follows from Lemma~\ref{lem:rank_Y_aug}.
Accordingly,
we analyze the condition 
 $ \tilde A^t_\kappa\, \Psi^{(\mathcal{Y})}(x) = 0$, 
where $\Psi^{(\mathcal{Y})}(x)$ and $\tilde A^t_\kappa$ are defined in
(\ref{eq:psi_processive}--\ref{eq:def_Ak_aug}), respectively.

The underlying graph of the translated
network~(\ref{eq:augmented_network_proc}) consists of a single
connected component that is strongly connected, so we obtain the 
following consequence of 
  \cite[Lemma~2]{TG}.
\begin{corollary} \label{cor:rk_A-kappa}
  The $(2n+2) \times (2n+2)$-matrix $\tilde A^t_\kappa$ from
  (\ref{eq:def_Ak_aug}) has rank $2n+1$. Moreover, $\ker(\tilde
  A^t_\kappa)$ is spanned by a positive vector $\rho$, whose entries
  are rational functions of the $k_i$ and $\ell_i$.
\end{corollary}

  \begin{remark}
    In principle one may explicitly compute the the elements of $\rho$
    by using the Matrix-Tree Theorem. To establish uniqueness of
    steady states (the aim of this section), however, one needs only existence
    of a positive vector spanning $\ker(\tilde A_\kappa^t)$,
    which is given by Corollary~\ref{cor:rk_A-kappa}. As the explicit
    computation of the $\rho_i$ is a rather tedious process, we omit
    this here. The interested reader is referred to
    Appendix~\ref{sec:coefficients}, where we comment on the computation
    of the vector $\rho_i$ in some detail.
  \end{remark}

  Following Corollary \ref{cor:rk_A-kappa}, we let $\rho\in\R_+^{2n+2}$ be a
  vector that spans $\ker(\tilde A_\kappa^t)$. Thus, by~\eqref{eq:ss_kerAk}
of the chemical reaction system defined by the processive
network~\eqref{eq:network_processive} if and only if 
  \begin{displaymath}
    \Psi^{(\mathcal{Y})}(x) = \alpha\, \rho \quad {\rm for~some~}
    \alpha >0~, 
  \end{displaymath}   
  where $\Psi^{(\mathcal{Y})}(x)$ is defined in
  (\ref{eq:psi_processive}). 
  In other words: 
  \begin{align}
    x_1\, x_3 &= \alpha\, \rho_1 \label{eq:monomial-1}\\ 
    x_{2i+3} &= \alpha\, \rho_{i+1} \quad {\rm for~} 1\leq i\leq n \label{eq:monomial-2} \\
    x_2\, x_4 &= \alpha\, \rho_{n+2} \label{eq:monomial-3} \\
   x_{2n+6-2i} &= \alpha\, \rho_{n+2+i}  \quad {\rm for~} 1\leq i\leq n ~. \label{eq:monomial-4}
  \end{align}
To eliminate $\alpha$, we divide equations~\eqref{eq:monomial-1}
and~\eqref{eq:monomial-2} by $x_6 =\alpha\, \rho_{2n+2}$ (the $i=n$
case of~\eqref{eq:monomial-4}) and divide
equations~\eqref{eq:monomial-3} and~\eqref{eq:monomial-4} by $x_{2n+3}
= \alpha\,    \rho_{n+1}$ (the $i=n$ case of~\eqref{eq:monomial-2}).
We thereby obtain the following implicit equations defining the set of
steady states:
\begin{align}
  \label{eq:fraction-1}
  \frac{x_1\, x_3}{x_6} &= \frac{\rho_1}{\rho_{2n+2}} \\
  \label{eq:fraction-2}
  \frac{x_{2i+3}}{x_6} &= \frac{\rho_{i+1}}{\rho_{2n+2}} \quad {\rm
    for~} 1\leq i\leq n \\
  \label{eq:fraction-3}
  \frac{x_2\, x_4}{x_{2n+3}} &= \frac{\rho_{n+2}}{\rho_{n+1}} \\
  \label{eq:fraction-4}  
  \frac{x_{2n+6-2i}}{x_{2n+3}} &=   \frac{\rho_{n+2+i}}{\rho_{n+1}}
  \quad {\rm for~} 1\leq i\leq n~.
\end{align}

These steady state equations are binomials in the $x_i$'s (for instance,
$x_1x_3 - \frac{\rho_1}{\rho_{2n+2}}x_6 = 0$), i.e., the processive systems
have {\em toric steady states} \cite{TSS}, just like the distributive
systems. Therefore, following~\cite[Theorem 3.11]{TSS}, we obtain
the following parametrization of positive steady states in terms of
the coordinates of $\rho$ and the free variables $x_2$, $x_3$, and $x_6$: 

\begin{proposition}[Parametrization of the steady states of
  the processive network] \label{prop:parametrization}
  Let $n$ be a positive integer.
  The set of positive steady states of the chemical reaction system
  defined by the processive $n$-site
  network~\eqref{eq:network_processive} and any choice of rate constants is three-dimensional and is
  the image of the following map $\chi = \chi_{n, \{k_i,\ell_i  \}}$: 
  \begin{align*}
    \chi : \mathbb{R}^3_{+} & \to \mathbb{R}^{2n+4}_{+} \\
    \chi(x_2, x_3, x_6) &:=(x_1, x_2, \dots, x_{2n+4})
  \end{align*}  
  given by
  \begin{align*}
    x_1 &:=  \frac{\rho_1}{\rho_{2n+2}} \, \frac{x_6}{x_3}
    & \quad x_4 &:=  
      \frac{\rho_{n+2}}{\rho_{2n+2}}\, \frac{x_6}{x_2} 
\\
    x_{2i+3} &:=  \frac{\rho_{i+1}}{\rho_{2n+2}} \,  x_6 ,\; {\rm
      for~}1\leq i\leq n 
    & \quad  
    x_{2i+6} &:=  
       \frac{\rho_{2n+2-i}}{\rho_{2n+2}}\, \, x_6,\; {\rm for~} 1\leq i\leq n-1. 
  \end{align*}
\end{proposition}
\begin{proof}
  The expressions for $x_1$ and $x_{2i+3}$ follow from 
  equations~\eqref{eq:fraction-1} and~\eqref{eq:fraction-2}, respectively.
  The expression for $x_4$ follows from equation~\eqref{eq:fraction-3}
  together with the equation 
  \begin{align} \label{eq:x_2n+3}
    x_{2n+3} =  
      \frac{\rho_{n+1}}{\rho_{2n+2}}\, x_6~,
  \end{align}
  which in turn follows 
    from
    the $i=n$ case of equation~\eqref{eq:fraction-4}.
  The expression for $x_{2i+6}$
  follows from equations~\eqref{eq:fraction-4} and~\eqref{eq:x_2n+3}
  again, together with an index shift that replaces $i$ (where $1 \leq
  i \leq n-1$) by $n-i$ (so, $2n+6-2i \mapsto 2i+6$ and $n+2+i \mapsto
  2n+2-i$).
\end{proof}

\begin{remark}
  That we could achieve Proposition~\ref{prop:parametrization} was
  guaranteed by the rational parametrization theorem for multisite
  systems of Thomson and Gunawardena~\cite{TG}; see also \cite[Theorem
  3.11]{TSS}.
      An alternative derivation follows from a recent result of Feliu and Wiuf~\cite[Theorem~1]{FeliuWiuf}. 
    This result guarantees that one may express the concentrations of
    the \lq intermediate\rq{} species $S_0K$, \ldots, $S_{n-1} K$
    and hence the variables $x_5$, $x_7$, \ldots, $x_{2n+3}$ in terms
    of the product $x_1\, x_3$. Likewise one may express the concentrations of
    the \lq intermediate\rq{} species $S_1 F$, \ldots, $S_{n} F$
    and hence the variables $x_6$, $x_8$, \ldots, $x_{2n+4}$ in terms
    of the product $x_2\, x_4$. By exploiting the steady state
    relation of $x_1\, x_3$ and $x_2\, x_4$ one may then arrive at a
    parameterization.  
	Although the approach we took is more lengthy, it allows us to see that Johnston's analysis of the 1-site network generalizes~\cite{translated}.
\end{remark}

\begin{remark}
In the parametrization in Proposition~\ref{prop:parametrization}, two of the coordinates require dividing by $x_2$ or $x_3$, so the parametrization is not technically a monomial map.  However, this can be made monomial easily: by introducing $y:=\frac{x_6}{x_2x_3}$, so that the parametrization accepts as input $(x_2,x_3,y)$, we see that $\frac{x_6}{x_3}$ is replaced by $x_2y$ and $\frac{x_6}{x_2}$ is replaced by $x_3y$.
\end{remark}

Below we will restate Proposition~\ref{prop:parametrization} so that
we can apply results from~\cite{signs} to rule out multistationarity.
First we require some notation. 

\noindent
{\bf Notation.}
\begin{itemize}
\item For $x,y \in \R^n$,
we denote the componentwise (or Hadamard) product by $x \circ y \in \R^n$, that is, $(x \circ y)_i = x_i y_i$.
\item For $x \in \R^n_+$, the vector $\ln(x) \in \R^n$ is defined componentwise: $\ln(x)_i:=\ln(x_i)$.
\item For a vector $x \in \R^n$, we obtain the \emph{sign vector}
      $\sign(x)\in \{-,0,+\}^n$ by applying the sign function
      componentwise. For a subset $X$ of $\R^n$, we then have  $\sign(X):=\{ \sign(x) \mid x \in X\}$.
\end{itemize}

We collect the exponents of $x_2$, $x_3$, and $x_6$ in the above 
parametrization (Proposition~\ref{prop:parametrization}) as rows 
of a $3 \times (2n+2)$-matrix we call $B$:
\begin{equation}
  \label{eq:exp_mat}
  B^t 
  := \left[
    \begin{array}{rrrr | rrrr}
      0 & 1 & 0 & -1 & 0 & 0 & \cdots & 0 \\
      -1 & 0 & 1 &  0 & 0 & 0 & \cdots & 0 \\
      1 & 0 & 0 &  1 & 1 & 1 & \cdots & 1 \\
    \end{array}
  \right]~.
\end{equation}
  Also, we use $x^*$ to denote the value of $\chi$ at $(1,1,1)$:
  \begin{equation} \label{eq:x-star}
    x^*= x^*(n, \{k_i,\ell_i\}):=\chi(1,1,1) \in \mathbb{R}_{>0}^{2n+2} ~.
  \end{equation}
  We obtain the following representation of the map $\chi(\cdot)$ from
  Proposition~\ref{prop:parametrization}:
\begin{proposition}[Parametrization, restated] \label{prop:param_restated}
  Let $\xi = (\xi_1,\xi_2,\xi_3)$ be a vector of indeterminates. 
  Then the map $\chi$ given in Proposition~\ref{prop:parametrization}
  can be rewritten as:
  \begin{align*}
    \chi(\xi) = x^* \circ \Psi^{(B)}(\xi)~, 
  \end{align*}
  where the matrix $B$ and the vector $x^*$ are given 
  in~(\ref{eq:exp_mat}--\ref{eq:x-star}) 
  {and $\Psi^{(B)}(\xi)$ is as
    in (\ref{eq:Psi_def}).
  } 
  Thus, distinct positive vectors $a, b \in \mathbb{R}^{2n+2}_+$ are
  both steady states of the system  
  if and only if $\ln   b - \ln a \in \im(B)$.
\end{proposition}

\begin{proof}
  Follows from Proposition~\ref{prop:parametrization},
  the construction of both, $B$ and $x^*$ and the fact that by
    (\ref{eq:Psi_def}) one has
     $\Psi^{(B)}(\xi) = (\xi^{b_1},\, \ldots,\, \xi^{b_{2n+2}})^t$.
\end{proof}

Next we consider steady states within a stoichiometric compatibility
class, that is, we analyze the intersection of ${\rm Im}(\chi)$ with
parallel translates $x^\prime+ \St$ of the stoichiometric subspace of
the processive network~\eqref{eq:network_processive}.
The following is an application of the
discussion preceding \cite[Proposition~3.9]{signs}. 
For any $x^\prime \in \R^{2n+2}_+$, the intersection ${\rm Im}(\chi)
\cap (x^\prime+ \St) $ is nonempty 
if and only if there 
exist vectors $\xi \in \mathbb{R}^3_+$ and $u\in \St$ 
such that
\begin{align} \label{eq:steady-state-in-class}
  \chi(\xi) = x^\prime + u.
\end{align}
Let $\mathcal{A} \in \mathbb{R}^{3 \times (2n+2)}$ be 
  the (full-rank) 
matrix with ${\rm ker} \mathcal{A} = \St$ 
given in
  equation~(\ref{eq:matrix_A}) of Remark~\ref{rmk:cons_laws}.
Then, equation~\eqref{eq:steady-state-in-class} implies that
\begin{displaymath}
  \mathcal{A}\, \chi(\xi) = \mathcal{A}\, x^\prime.
\end{displaymath}
Therefore, if the map $f_{x^*} : \R^3_+ \to \R^3$  given by 
\begin{equation} \label{eq:phi-for-injectivity}
  f_{x^*}(\xi) ~:=~ \mathcal{A}\, \chi(\xi)
\end{equation}
is injective, then every parallel translate $x^\prime
+ \St$ (and thus, every stoichiometric compatibility class) contains
at most one element of ${\rm Im}(\chi)$.  So, by Propositions~\ref{prop:parametrization} and~\ref{prop:param_restated}, multistationarity would be precluded for all  processive systems.
To decide injectivity of $ f_{x^*}$, we use the following result which 
is a direct consequence of \cite[Proposition~3.9]{signs}:
\begin{proposition}[M\"uller {\em et al.}]\label{prop:signs}
  Let $\St$ be the stoichiometric subspace of the processive
  network~\eqref{eq:network_processive}, and let $B$ be as
  in~(\ref{eq:exp_mat}). If 
  \begin{displaymath}
    \sign(\im(B)) \cap \sign(\St) = \{ 0 \}~,
  \end{displaymath}
  then the polynomial map $f_{x^*} : \R^3_+ \to \R^3$ given  in~\eqref{eq:phi-for-injectivity} is injective.
\end{proposition}
\begin{proof}
  Follows from the equivalence (ii) $\Leftrightarrow$ (iii) of
  \cite[Proposition~3.9]{signs}.
\end{proof}

\begin{remark} \label{rmk:signs-result}
Proposition~\ref{prop:signs} appears in many works, for instance, \cite{BP,ShinarFeinberg2012}.  In fact, criteria for injectivity (including those given by sign conditions) have a long history in the study of reaction systems.  For a more detailed discussion, see~\cite{signs}.
\end{remark}


We can now give the main result of this section:
  \begin{theorem} \label{thm:existence-uniqueness-via-signs}
  Let $n$ be a positive integer. 
For any chemical reaction system~\eqref{eq:ODE} arising from the processive $n$-site network~\eqref{eq:network_processive} and any choice of rate constants,
 each stoichiometric compatibility class $\invtPoly$ contains a unique steady state $\eta$, 
and $\eta$ is a positive steady state.
\end{theorem}
\begin{proof}
  As explained earlier in Remark~\ref{rmk:brouwer}, 
  the existence of at least one (necessarily positive) steady state in
  $\invtPoly$ is guaranteed by the Brouwer fixed-point theorem.
  Thus, to prove the theorem, we need only preclude multistationarity. 
  So, by Proposition~\ref{prop:signs} and the preceding discussion, it
  suffices to prove the nonexistence of nonzero vectors 
  $\alpha\in\im(B)$ and $s\in\St$ with $\sign(s) =
  \sign(\alpha)$. 
  We begin by defining $\tilde B$ as the matrix obtained from $B$
  by adding the first two columns to the third column, so $\im(B) =
  \im(\tilde B)$:
  \begin{align} \label{eq:B}
    \tilde B^t = \left[
      \begin{array}{rrrr | rrrr}
        0 & 1 & 0 & -1 & 0 & 0 & \cdots & 0 \\
       -1 & 0 & 1 &  0 & 0 & 0 & \cdots & 0 \\
        0 & 1 & 1 &  0 & 1 & 1 & \cdots & 1 \\
      \end{array}
    \right]~.
    \end{align}
    We proceed by contradiction: assume that there exist nonzero vectors 
    $\alpha\in\im(\tilde B)$ and $s\in \St$ with 
    $\sign(\alpha) = \sign(s)$.
    By~\eqref{eq:B}, $\alpha\in\im(\tilde B)$ implies that 
    \begin{displaymath}
      \sign(\alpha_5) = \sign(\alpha_6) = \dots = \sign(\alpha_{2n+4}),
    \end{displaymath}
    so we conclude that 
    \begin{equation}
      \label{eq:s_equal}
      \sign(s_5) = \sign(s_6) = \dots = \sign(s_{2n+4})
    \end{equation}
    as well. Also, from our choice of $\mathcal{A}$ in~\eqref{eq:matrix_A}, the vector $s\in\ker(\mathcal{A})$ satisfies
    \begin{align*}
      s_1 = -s_5 - s_7 - \cdots - s_{2n+3} \\ 
      s_2 = -s_6 - s_8 - \cdots - s_{2n+4} \\ 
      \;
      s_3 + s_4 = -s_5 - s_6 - \cdots - s_{2n+4}~. 
    \end{align*}
 Thus, using (\ref{eq:s_equal}), the coordinates $s_1$, $s_2$, \ldots, $s_5$
    must satisfy 
    \begin{equation}
      \label{eq:s_diff}
      \sign(s_1)= -\sign(s_5), \quad
      \sign(s_2) = -\sign(s_5), \quad 
      \sign(s_3+s_4)
      = -\sign(s_5).
    \end{equation}
    We assumed that $\sign(s) = \sign(\alpha)$, so the coordinates 
    $\alpha_1$, $\alpha_2$, \ldots, $\alpha_5$ must satisfy the same
    conditions~(\ref{eq:s_diff}). 
    Now we make use of a vector $\beta \in \mathbb{R}^3$ for which
    $\alpha=\tilde B\, \beta$, which exists because $\alpha \in \im(B)
    = \im(\tilde B)$.  By~\eqref{eq:B}, we have:
    \begin{align*}
      \alpha_1 = - \beta_2 \quad 
      \alpha_2 = \beta_1+ \beta_3 \quad
      \alpha_3 = \beta_2 + \beta_3\quad
      \alpha_4  = -\beta_1\quad
      \alpha_5 = \beta_3~.
    \end{align*}
    Thus, the conditions on $\alpha$ arising from~(\ref{eq:s_diff}) imply: 
    \begin{align} \label{eq:beta-inequalities}
      \sign(-\beta_2)= \sign(- \beta_3) \quad
      \sign(\beta_1 + \beta_3)= \sign(- \beta_3) \quad
      \sign(-\beta_1 + \beta_2+ \beta_3)= \sign(- \beta_3)~.
    \end{align}
    We distinguish three cases based on the sign of $\beta_2$.

    \noindent 
    {\bf Case One}: $\beta_2>0$.  The conditions~\eqref{eq:beta-inequalities} yield 
    \begin{displaymath}
      -\beta_2<0, \quad
      -\beta_3<0, \quad      
      \beta_1+\beta_3<0, \quad
      -\beta_1+\beta_2+\beta_3<0~.
    \end{displaymath}
    The sum of the first, second, and fourth inequalities yields the
    consequence $-\beta_1<0$, while the sum of the second and the
    third inequalities yields the consequence $  \beta_1 < 0$, which
    is a contradiction. 
    
    \noindent
    {\bf Case Two}:  $\beta_2<0$. This similarly yields a
    contradiction (reverse all inequalities in Case One). 

    \noindent
    {\bf Case Three}:  $\beta_2 = 0$.  The first condition 
    in~\eqref{eq:beta-inequalities} implies $\beta_3 = 0$, which, by the 
    second condition in~\eqref{eq:beta-inequalities}, implies that 
    $\beta_1=0$.  Thus, $\alpha=\tilde B\, \beta $ is zero ,
    so we again reach a contradiction. 
  \end{proof}

Having established the existence and uniqueness of steady states, the
next section addresses the natural next question: global convergence.

\section{Convergence to a global attractor}\label{sec:convergence}
In this section, we prove that 
each steady state of the processive network taken with mass-action kinetics is a global attractor of the corresponding compatibility class (Theorem~\ref{thm:global_conv}).  
The proof is via Lemma~\ref{lem:AS}, which is due to Angeli and Sontag~\cite{AS}.  Their work is one of many recent papers proving convergence of reaction systems by way of monotone systems theory; see  Angeli, De Leenheer, and Sontag~\cite{ADS10}, Banaji and Mierczynski~\cite{BanajiM}, and Donnell and Banaji~\cite{DB}.

{\bf Setup.}  We begin by recalling the setup in Angeli and Sontag~\cite[\S3]{AS}.  We consider any reaction kinetics system with $s$ chemical species and $m$ reactions (where each pair of reversible reactions is counted only once) given by 
$\dot{x} = \Gamma \, R(x)$, as in~\eqref{eq:ODE}.  Each such system together with a vector $\sigma \in \Rnn^s$ (viewed as an initial condition of~\eqref{eq:ODE}) defines 
 another ODE system:
\begin{align} \label{eq:R-system}
\dot{c} = f_{\sigma}(c) := R(\sigma + \Gamma c),
\end{align}
with associated state space  
(which is sometimes called the space of ``reaction coordinates'')
\begin{align} \label{eq:X-sigma-state-space}
X_{\sigma} = \left\{ c \in \R^m \mid \sigma + \Gamma c \in \Rnn^s \right\} .
\end{align}

To state Lemma~\ref{lem:AS} below, we require the following
definition.

\begin{definition} \label{def:monotone} ~
\begin{enumerate}
\item 
	The nonnegative orthant $\mathbb{R}^m_{\geq 0}$ defines
	a partial order on $\mathbb{R}^m$ given by $c_1 \succcurlyeq c_2$ if 
	$c_1 - c_2 \in \mathbb{R}^m_{\geq 0}$.  
	Also, we write $c_1 \succ c_2$ if $c_1 \succcurlyeq c_2$ with $c_1 \neq c_2$,
	and 
	$c_1 \gg c_2$ if $c_1 - c_2 \in \mathbb{R}^m_{>0}.$ 
\item A dynamical system with state space 
$X \subseteq \mathbb{R}^m$ 
and flow denoted by $\phi_t(c)$ (for initial condition $c$) is 
{\em monotone with respect to the nonnegative orthant} $\mathbb{R}^m_{\geq 0} $
if 
the partial order arising from $\mathbb{R}^m_{\geq 0}$
is preserved 
  by the forward flow:
 for $c_1, c_2 \in X$, if $c_1 \geq c_2$ then $\phi_t(c_1) \geq \phi_t(c_2)$ for all $t \geq 0$.  
A dynamical system is 
{\em strongly monotone with respect to the nonnegative orthant}  
if 
it is monotone with respect to the nonnegative orthant 
and, additionally, for $c_1, c_2 \in X$, the relation $c_1 \succ c_2$ implies that $\phi_t(c_1) \gg \phi_t(c_2)$ for all $t >0$.
\end{enumerate}
\end{definition}

The following result is due to Angeli and Sontag \cite[Corollary 3.3]{AS}.
\begin{lemma}[Angeli and Sontag] \label{lem:AS}
Let $R$, $\Gamma$, and $\sigma$ be as in the setup above.  Assume that:
\begin{enumerate}
\item the stoichiometric matrix $\Gamma$ has rank $m-1$, with kernel spanned by some positive vector (i.e., in $\Rplus^m$),
\item every trajectory of the reaction kinetics system~\eqref{eq:ODE} is bounded, and 
\item the system $\dot c = f_{\sigma}(c)$ defined in~\eqref{eq:R-system} is strongly monotone
with respect to the nonnegative orthant.
\end{enumerate}
Then there exists a unique $\eta=\eta_{\sigma} \in \Rnn^s$ such that for any initial condition $\mu \in \Rnn^s$ that is stoichiometrically compatible with $\sigma$ (i.e., $\mu - \sigma \in {\rm Im}(\Gamma)$), the trajectory $x(t)$ of the reaction kinetics system~\eqref{eq:ODE} with initial condition $x(0)= \mu$ converges to $\eta$: $\lim\limits_{t \to \infty} x(t) = \eta$.
\end{lemma}

Following closely the example of the 1-site system presented by Angeli and Sontag~\cite[\S3]{AS}, we now use Lemma~\ref{lem:AS} to extend their result beyond the $n=1$ case: the following result states that the processive $n$-site network~\eqref{eq:network_processive} is convergent. 
Note that in applying Lemma~\ref{lem:AS}, we will show that the new system in~\eqref{eq:R-system}, not the original processive system, is strongly monotone. 
Also note that by obtaining existence and uniqueness of steady states, the theorem  
supersedes our earlier result (Theorem~\ref{thm:existence-uniqueness-via-signs}), but the approach here can not obtain the parametrization of the steady states  we accomplished earlier (Proposition~\ref{prop:parametrization}).

\begin{theorem} \label{thm:global_conv}
Let $n$ be a positive integer. 
For any chemical reaction system~\eqref{eq:ODE} arising from the processive $n$-site network~\eqref{eq:network_processive} and any choice of rate constants,
\begin{enumerate}
	\item each stoichiometric compatibility class $\invtPoly$ contains a unique steady state $\eta$, 
	\item $\eta$ is a positive steady state, and 
	\item $\eta$ is the global attractor of $\invtPoly$.
\end{enumerate}
\end{theorem}
\begin{proof}
Let $\sigma \in \invtPoly$.  
The result will follow from Lemma~\ref{lem:AS} applied to this reaction system and the vector~$\sigma$, once we verify its three hypotheses.  

  For hypothesis (1) we note that the rank of $\Gamma$ is $(2n+2)-1$ by
  Lemma~\ref{lem:rank-gamma}.

For hypothesis~(2) of Lemma~\ref{lem:AS}, every stoichiometric compatibility class is bounded due to the conservation laws
(cf.\ Remark~\ref{rmk:cons_laws}).
Thus, trajectories of~\eqref{eq:ODE} are bounded.  

Finally, we must verify that the system~\eqref{eq:R-system} is strongly monotone.  We begin by showing that it is monotone with respect to the nonnegative orthant.  It suffices (by Proposition 1.1 and Remark~1.1 in~\cite[\S 3.1]{SmithBook}) to show that the Jacobian matrix of $f_{\sigma}(c):=R(\sigma + \Gamma c)$ with respect to $c$ has nonnegative off-diagonal entries for all $c \in X_{\sigma}$.  
Note that this reaction rate function $R$ appeared earlier in~\eqref{eq:def_R_proc}.
For simplicity, we introduce $z:=\sigma + \Gamma c$, so by the chain rule, the 
Jacobian matrix of $f_{\sigma}(c):=R(\sigma + \Gamma c)$ with respect to $c$ is
${\rm Jac}_c f_{\sigma}(c)  = {\rm Jac}_x R(z)
\, \Gamma$, which
from~(\ref{eq:def_Gamma_proc}--\ref{eq:def_R_proc}) is  
 the following $(2n+2) \times (2n+2)$-matrix:

\begin{align}
\label{eq:Jacobian_any_n}
    \left[  
    \begin{array}{c}
    \left( -k_1 \left( z_3+ z_1 \right) - k_2 \right) e_1 + k_2 e_2 + k_1 z_3 e_{n+1} +k_1 z_1 e_{2n+2} \\
        \hline
    k_3 e_1 - (k_3 +k_4) e_2 + k_4 e_3 \\
    k_5 e_2 - (k_5 +k_6) e_3 + k_6 e_4 \\
    \vdots \\
    k_{2n-1} e_{n-1} - (k_{2n-1} +k_{2n}) e_n + k_{2n} e_{n+1} \\
        \hline
    k_{2n+1} e_n - k_{2n+1} e_{n+1} \\
        \hline
    \ell_{2n+1} z_2 e_{n+1} +
    \left( -\ell_{2n+1} \left( z_4+ z_2 \right) - \ell_{2n} \right) e_{n+2} +
    \ell_{2n} e_{n+3} + \ell_{2n+1} z_4 e_{n+2}  \\
        \hline
    \ell_{2n-1} e_{n+2} - (\ell_{2n-1} +\ell_{2n-2}) e_{n+3} + \ell_{2n-2} e_{n+4} \\
    \ell_{2n-3} e_{n+3} - (\ell_{2n-3} +\ell_{2n-4}) e_{n+4} + \ell_{2n-4} e_{n+5} \\
    \vdots \\
    \ell_{3} e_{2n} - (\ell_{3} +\ell_{2}) e_{2n+1} + \ell_{2} e_{2n+2} \\
        \hline
    \ell_{1} e_{2n+1} - \ell_{1} e_{2n+2}
    \end{array}
    \right]~.
\end{align}
  
By inspection of the Jacobian matrix~\eqref{eq:Jacobian_any_n}, each nonzero off-diagonal entry either is some $\ell_i$ or $k_j$, which is strictly positive, or has the form 
$k_j z_i$ or $\ell_j z_i$ (for some $i$) 
and such a term is nonnegative for $c \in X_{\sigma}$
(recall that the 
system~\eqref{eq:R-system} 
evolves on the space $X_{\sigma}$ defined in~\eqref{eq:X-sigma-state-space}, so 
$z=\sigma + \Gamma c \in \Rnn^{2n+2}$. 
  
Now we show that the system~\eqref{eq:R-system} is strongly monotone by checking that the Jacobian matrix~\eqref{eq:Jacobian_any_n} is almost everywhere irreducible along trajectories of~\eqref{eq:R-system} (see Theorem 1.1 of~\cite[\S 4.1]{SmithBook}), i.e., that the matrix is almost everywhere the adjacency matrix of a strongly connected directed graph. 
By inspection of~\eqref{eq:Jacobian_any_n}, this directed graph always contains the edges $1 \leftrightarrow 2 \leftrightarrow  \cdots \leftrightarrow n+1$ and $n+2 \leftrightarrow n+3 \leftrightarrow \cdots \leftrightarrow 2n+2$ (because $k_i, \ell_i >0$ for all $i$), and the only possible edges between these two components are $1 \to 2n+2$ and $n+2 \to n+1$, so we must show that the corresponding two entries in the matrix~\eqref{eq:Jacobian_any_n}, namely $k_1 z_1 = k_1 (\sigma_K-c_1+c_{n+1})$ and $\ell_{2n+1} z_2 = \ell_{2n+1} (\sigma_F-c_{n+2}+c_{2n+2})$, are almost everywhere nonzero along trajectories.
  
By symmetry between $K$ and $F$, we need only verify the first case.  
We proceed by contradiction: assume that $z_1(t) = \sigma_K-c_1(t)+c_{n+1}(t) \equiv 0$ for a positive amount of time $t$ along a trajectory $c(t)$ of~\eqref{eq:R-system}.  So, 
using~\eqref{eq:def_R_proc}, 
this subtrajectory satisfies:
	\begin{align*}
	0 \equiv \dot c_1(t) - \dot c_{n+1}(t) = (0 - k_{2}z_{5}(t) ) - k_{2n+1} z_{2n+3}(t) ~.
	\end{align*}

But, $z_{5} (t)  \geq 0$ and $z_{2n+3} (t) \geq 0$, so both must equal zero for the above to hold.  Additionally, we conclude that $\dot c_1(t) \equiv 0$ and $\dot c_{n+1}(t) \equiv 0$.
Hence, the base case is complete for showing by induction on $i=0,1,\dots,n-1$ that
	\begin{align} \label{eq:zeros_we_get}
	    z_{2i+5} (t)  \equiv 0 \quad {\rm and }  \quad \dot c_{i+1}(t) \equiv 0~.
	\end{align}
For the $i$-th step, we use the inductive hypothesis
(namely, $z_{2i+3}(t)  = \sigma_{2i+3} + c_i(t) - c_{i+1}(t) \equiv 0$ and $\dot c_{i}(t) \equiv 0$)
to obtain:
\begin{align*}
0 \equiv \dot c_i(t) - \dot c_{i+1}(t) = 0 - 
	\left( 0 - k_{2i+2} z_{2i+5}(t) \right) ~.
\end{align*}
Thus, the desired equalities~\eqref{eq:zeros_we_get} hold. 
Hence,
\begin{align} \label{eq:long_sum}
0 = 0+ \cdots +0 &\equiv z_1(t) +z_{5}(t) + z_{7}(t) + \cdots + z_{2n+3}(t) \notag \\
    &= 
    \left( \sigma_{1} - c_{1} + c_{n+1}  \right) +
    \left(   \sigma_{5} + c_{1} - c_{2} \right) +
    \cdots +
    \left(  \sigma_{2n+3} + c_{n} - c_{n+1} \right) \notag \\
    &=
    \sigma_1 + (\sigma_{5} + \sigma_7 + \cdots + \sigma_{2n+3} ) >0~,
\end{align}
where the inequality in~\eqref{eq:long_sum} follows because the  sum in~\eqref{eq:long_sum} represents the total 
(free and bound)
amount of kinase $K$ present in the initial condition $\sigma$, which must be strictly positive in order for $\sigma \in \invtPoly$ 
(recall Remark~\ref{rmk:cons_laws}).
  Thus, we obtain a contradiction, and this completes the proof.
%
\end{proof}

\section{Discussion}
\label{sec:discussion}
In this section, we comment on related works.  
The following three remarks highlight alternative methods to the one taken here for precluding multistationarity in processive networks.
For an overview of known methods for assessing multistationarity in reaction kinetics systems, see the introduction of~\cite{simplifying}.  
For a historical survey of experimental and theoretical findings concerning multistationarity, see the book of Marin and Yablonsky~\cite[Chapter 8]{MY}.  


\begin{remark} \label{rmk:DSR}
Readers familiar with ``directed species-reaction graphs'' (DSR graphs) can verify that the DSR graph arising from the processive multisite network~\eqref{eq:network_processive} satisfies Banaji and Craciun's condition (*) in \cite{BanajiCraciun2009} and thereby conclude that processive systems do not admit multistationarity. 
\end{remark}

\begin{remark} \label{rmk:def_one_alg}
Another approach to ruling out multistationarity in processive systems is via the Deficiency One Algorithm due to Feinberg. 
Namely, one could apply 
a criterion for multistationarity of regular deficiency-one networks~\cite[Corollary 4.1]{Fein95DefOne}, determine that the resulting system of inequalities is infeasible, and then conclude that multiple steady states are precluded.  Indeed, for small $n$, this can be verified by the CRN Toolbox software~\cite{Toolbox}.
\end{remark}

\begin{remark} \label{rmk:F_Wiuf}
A third approach to analyzing processive systems is to use the recent work of Feliu and Wiuf~\cite{FeliuWiuf}.  Namely, in their notation, each processive $n$-site network~\eqref{eq:network_processive} is an ``extension model'' of the following ``core model'' network:
$$ S_0 + E \rightarrow S_n + E \quad \quad S_n + F \rightarrow S_0 + F~. $$
The corresponding ``canonical model'' is obtained by adding the reactions $X \rightleftharpoons S_0+E $ and $S_n+F \rightleftharpoons Y $.  This
canonical model
can be determined to preclude multistationarity, via the CRN Toolbox software~\cite{Toolbox} (which applies the Deficiency One Algorithm~\cite{Fein95DefOne} in this case) or the online software tool CoNtRol~\cite{control} (which applies injectivity criteria of Banaji and Pantea~\cite{BP}).  
Corollary~6.1 in the work of Feliu and Wiuf states that if a canonical model admits at most $N$ steady states, then every extension model of the core model also admits no more than $N$ steady states.
So, that corollary allows us to conclude 

that the entire family of processive $n$-site networks~\eqref{eq:network_processive} also precludes multistationarity.
Also, their results can give information about the stability of the resulting unique steady states.  However, 
even if we could readily apply Proposition 2 in the Data Supplement of that work, we would obtain only local stability.  
In Section~\ref{sec:convergence}, we 
accomplished
 the stronger result of global stability by appealing to monotone systems theory.
\end{remark}
%

The next two remarks relate our convergence result to other such results.

\begin{remark} \label{rmk:DB}
As explained before Theorem~\ref{thm:global_conv}, our result extends the convergence result for 1-site systems due to Angeli and Sontag.  An alternate proof of convergence of the 1-site network is due to Donnell and Banaji~\cite[Example 3]{DB}, but their argument does not extend to $n$-site systems.
\end{remark}

\begin{remark} \label{rmk:convergence-results}
As mentioned earlier, Theorem~\ref{thm:global_conv} is one of many results proving the global convergence of various reaction systems by way of monotone systems theory~\cite{ADS10,BanajiM,DB}.  As a complement to monotone systems theory, other approaches to obtaining convergence theorems for reaction systems have been aimed at resolving the so-called Global Attractor Conjecture and related conjectures.  An overview of such recent results appears in~\cite[\S 1.1]{Anderson} and~\cite[\S 4]{GMS2}.  However, the aforementioned conjectures and results do not apply to the processive systems considered in our work.
\end{remark}
 
Finally, we identify other families of multisite systems for further study.
\begin{remark} \label{rmk:mixed_etc.}
As discussed earlier, many works have analyzed {\em distributive} multisite systems, in contrast with the {\em processive} versions analyzed in our work.
We now make note of two additional families of multisite systems.  The first is the class of {\em mixed} systems, in which the phosphorylation mechanism is processive and the dephosphorylation mechanism is distributive (or vice-versa); the $n=2$ case was considered in~\cite[\S 2.2]{ConradiUsing}.  
We conjecture that, like 
processive
 systems, 
mixed systems taken with mass-action kinetics admit a unique (positive) steady state in each stoichiometric compatibility class, and that this steady state is a global attractor.  The online software tool CoNtRol~\cite{control} verifies that for small $n$, steady states are unique because the systems satisfy certain injectivity criteria~\cite{BP}.  
As for convergence, the proof of Theorem~\ref{thm:global_conv} can not easily be modified to analyze mixed systems, so global convergence (if it holds) must be proved in another way.  We note that a related version of the mixed $2$-site system was considered by Gunawardena in~\cite{Guna}.  

A second potentially interesting class of networks arises when 
phosphorylation proceeds by a {\em semi-processive} mechanism~\cite[\S 4.2]{PM}, in which the kinase is capable of catalyzing the attachment of more than one phosphate group per binding event, but the maximum number of phosphate groups is not attached each time.  
Indeed, macromolecular crowding \cite{Ellis} causes distributive systems to function in a semi-processive manner~\cite{aoki}.  

We leave this class as a topic for future work. 
\end{remark}

\vskip .2in
\noindent \textbf{{\large Acknowledgments.}}
We thank  Murad Banaji and Pete Donnell for directing us to the relevant monotone systems literature.  We also thank Matthew Johnston for helpful discussions, 
	and two conscientious referees whose comments improved this work.

\bibliographystyle{amsplain}
\bibliography{multistationarity}

\appendix


\section{Obtaining the nullspace of $A^t_\kappa$ from
  (\ref{eq:augmented_network_proc})}
\label{sec:coefficients}

Here we focus on the nullspace of $A^t_\kappa$ and explain how it can
be obtained by studying the directed graph underling
network~(\ref{eq:augmented_network_proc}), given in
Fig.~\ref{fig:directed_graph} below.
\begin{figure}[!h]
  \centering
  \begin{minipage}{1.0\linewidth}
    \scalebox{0.8}{
      \begin{minipage}{1.0\linewidth}
        \begin{displaymath}
          \xymatrix@C=3ex{
            \overset{1}{\bullet} 
            \ar @<.4ex> @{-^>} [r] 
            &
            \ar @{-^>} [l] 
            \overset{2}{\bullet}
            &
            {\phantom{K}\hdots\phantom{K}}
            &
            \ar @<.4ex> @{-^>} [r] 
            &
            \ar @{-^>} [l] 
            \ar @<.4ex> @{-^>} [r] 
            &
            \ar @{-^>} [l] 
            &
            {\phantom{K}\hdots\phantom{K}}
            &
            \ar @<.4ex> @{-^>} [r] 
            &
            \ar @{-^>} [l] 
            \ar @<.4ex> @{-^>} [r] 
            &
            \ar @{-^>} [l] 
            &
            {\phantom{K}\hdots\phantom{K}}
            &
            \overset{n+1}{\bullet}
            \ar [d] ^{k_{2n+1}}
            \\
            \ar [u] ^{\ell_1}
            \underset{2n+2}{\bullet}
            &
            {\phantom{K}\hdots\phantom{K}}
            &
            \ar @<.4ex> @{-^>} [r] 
            &
            \ar @{-^>} [l] 
            \ar @<.4ex> @{-^>} [r] 
            &
            \ar @{-^>} [l] 
            &
            {\phantom{K}\hdots\phantom{K}}
            &
            \ar @<.4ex> @{-^>} [r] 
            &
            \ar @{-^>} [l] 
            \ar @<.4ex> @{-^>} [r] 
            &
            \ar @{-^>} [l] 
            &
            {\phantom{K}\hdots\phantom{K}}
            &
            \underset{n+3}{\bullet}
            \ar @<.4ex> @{-^>} [r] 
            &
            \ar @{-^>} [l] 
            \underset{n+2}{\bullet}
          }
        \end{displaymath}
      \end{minipage}
    }
  \end{minipage}
  \caption{Directed graph $G$ underlying the translated network~(\ref{eq:augmented_network_proc})}
  \label{fig:directed_graph}
\end{figure}
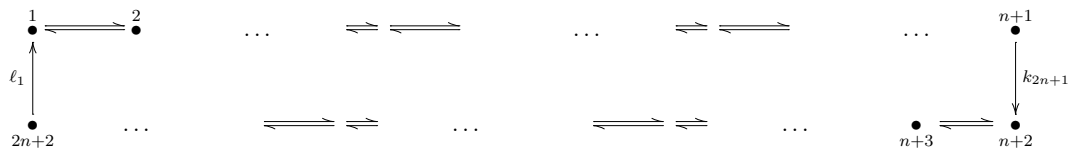

 
\medskip
\noindent
{\bf Notation} ($G^*$).  For a directed graph~$G$, we let $G^*$ denote
the undirected graph obtained from $G$ by making each directed edge
undirected (and allowing multiple edges in the resulting graph).
 
\begin{definition}[Directed spanning tree / spanning tree rooted at
  node $j$] \label{def:directed_sp_tree} \mbox{} \\
  Let $j$ be a node of a directed graph $G$.
  A subgraph $T$ is a {\em spanning tree} (of $G$) {\em rooted at $j$}, if it
  satisfies the following:
  \begin{enumerate}[{(}a{)}]
  \item $T$ contains all nodes of $G$,
  \item the undirected graph $T^*$ is acyclic and connected, and 
  \item for every node $v\neq j$ of $T$, there exists a directed path
    from $v$ to $j$.
  \end{enumerate}
  A subgraph is a {\em directed spanning tree} of $G$ if it is a
  spanning tree rooted at $j$, for some node $j$.
\end{definition}

\begin{remark}
  In a directed graph, a {\em sink} is a node that has no outgoing edges. 
  For 
  a spanning tree rooted at $j$, the unique
  sink is the node $j$. Any acyclic and connected subgraph that
  contains more than one sink is not a directed spanning tree.
\end{remark}

{
  Next we identify the directed spanning trees of $G$ from
  Fig.~\ref{fig:directed_graph}.  
  Note that $G$
  is cyclic, and due to the unidirectional edges labeled $k_{2n+1}$ and
  $\ell_1$, $G$ can be traversed in the clockwise direction
  only.
}

\begin{remark}[Acyclic, connected subgraphs of $G$ {from
    Fig.~\ref{fig:directed_graph}}] \label{rem:orientation} \mbox{} \\
  For a subgraph $T$ of $G$ that contains all nodes of $G$, 
  the undirected graph $T^*$ is acyclic and connected 
  if and only if $T$ satisfies the 
  following properties (cf.\ Fig.~\ref{fig:np_graph}):
  \begin{enumerate}[{(}i{)}]
  \item there is a unique node $p$ such that $T$ contains neither the
    edge $p \rightarrow p+1$ nor the edge $p \leftarrow p+1$ (where
    $p+1:=1$ if $p=2n+2$). 
  \item for all other nodes $q \neq p$, exactly one of the edges $q
    \rightarrow q+1$ and $q \leftarrow q+1$ is present in $T$. 
  \end{enumerate}
\end{remark}

\begin{figure}[!htb]
  \centering
  \includegraphics[width=0.8\linewidth]{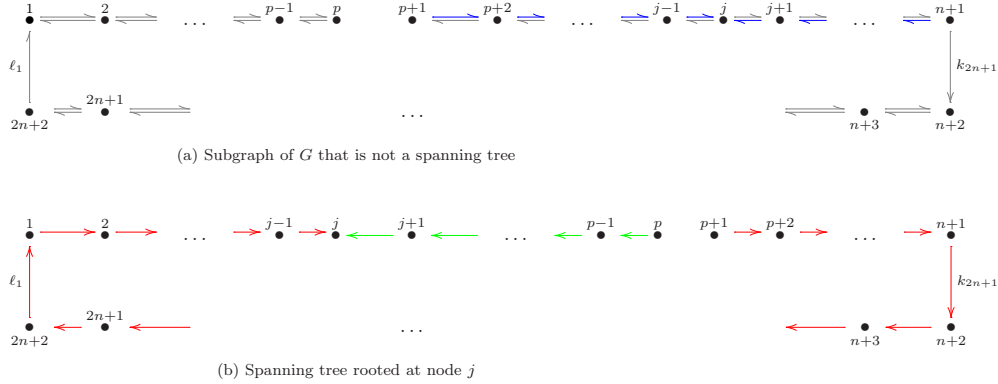}
  \caption{\label{fig:np_graph} (a) Subgraph obtained from $G$ by removing
    edges $p \to p+1$ and $p \leftarrow p+1$. 
    To obtain a subgraph $T$ for which the undirected graph $T^*$ is
    acyclic and connected, choose one edge from each gray pair of
    reversible edges. By choosing all the blue edges, one obtains two
    directed paths ending at $j$: one connecting the nodes $p+1$,\ldots, $j-1$ to $j$ and the
    other connecting  $j+1$, \ldots, $n+1$ to $j$. No choice of edges,
    however, will connect any of the following nodes to $j$: $n+2$,
    \ldots, $2n+2$ and $1$, \ldots, $p$.  Thus, any such subgraph
    will have at least two sinks. (b) Spanning tree $T_{j,p}$ (of $G$ from
    Fig.~\ref{fig:directed_graph}) rooted at $j$; this tree consists of two
    paths, one from $p$ to $j$ (green, counter-clockwise) and one from
    $p+1$ to $j$ (red, clockwise).} 
\end{figure}

Now we can determine the directed spanning trees of $G$ (recall
Definition~\ref{def:directed_sp_tree}): 
\begin{proposition}[Directed spanning trees of
  $G$ {from Fig.~\ref{fig:directed_graph}}] \label{prop:directed_sp_trees} \mbox{} \\
  For the directed graph $G$ in Fig.~\ref{fig:directed_graph},
  let $j$ and $p$ be integers such that
  \begin{align} \label{eq:condition_j_p}
    1 \leq j \leq p \leq n+1 \quad {\rm or } 
    \quad n+2 \leq j \leq p \leq 2n+2~.
  \end{align}
  Let $T_{j,p}$ be the subgraph of $G$ that contains all nodes of $G$
  and for which the edges are comprised of: 
  \begin{enumerate}[{(}1{)}]
  \item\label{item:2a} if $j \neq n+1, 2n+2$:
    \begin{enumerate}[{(}A{)}]
    \item the clockwise path from node $p+1$ to $j$, and
    \item the counter-clockwise path from $p$ to $j$ (cf.\
      Fig.~\ref{fig:np_graph}(b)).  
    \end{enumerate}
  \item\label{item:2b} if  $j=n+1$ or $j=2n+2$, the clockwise path
    from node $j+1$ to $j$ (where $j+1:=1$ if $j=2n+2$).
  \end{enumerate}
  Then $T_{j,p}$ is a directed spanning tree rooted at node
  $j$ that does not contain the edges $p \to p+1$ or $p \leftarrow
  p+1$ (where $p+1:=1$ if $p=2n+2$). Conversely, every spanning tree
  of $G$ has this form.
\end{proposition}

\begin{proof}
  Assume that $T_{j,p}$ is a subgraph as described in the proposition.
  By Definition~\ref{def:directed_sp_tree} and Remark~\ref{rem:orientation},
  it remains only to show that there exists a path from every node $v\neq
  j$ to $j$. Indeed, by points (\ref{item:2a}) and
  (\ref{item:2b}), every node belongs to a path that ends in $j$.
  
  Conversely, let $T$ be a spanning tree of $G$ rooted at $j$.  By 
  Remark~\ref{rem:orientation}, there exists a node $p$ such that 
  $T$ contains neither $p \to p+1$ nor $ p \leftarrow p+1$, so it 
  suffices to check that condition~\eqref{eq:condition_j_p} holds and the 
  edges of $T$ satisfy points~\eqref{item:2a} and~\eqref{item:2b}. 
  We first assume that $p$ violates
  condition~\eqref{eq:condition_j_p}. By symmetry between the two
  cases, we need only consider the case when $1 \leq j \leq n+1$ and
  $p\in \{1, \dots, j-1\} \cup \{n+2, \dots, 2n+2\}$. If $p\in \{1,
  \dots, j-1\}$, then there is no path in $T$ from $p$ to $j$;
  similarly, if $p \in \{n+2, \dots, 2n+2\}$, then there is no path
  from $n+2$ to $j$ (cf.\ Fig.~\ref{fig:np_graph}). Thus, $T$ is not a
  spanning tree rooted at $j$, which is a contradiction. Thus, $T$
  must satisfy condition~\eqref{eq:condition_j_p}, so it remains only
  to show that it must satisfy points~\eqref{item:2a} and
  \eqref{item:2b} as well.  Indeed in the first case (that is, if $j
  \neq n+1, 2n+2$), the paths (A) and (B) are the unique paths in $G$
  that do not use $p \to p+1$ to reach $j$ from $p+1$ and $p$,
  respectively, and all nodes except $j$ lie on exactly one of these
  paths, so the two paths comprise the edges of $T$.  Similarly, in
  the remaining case (if $j=n+1$ or $j=2n+2$), the clockwise path from
  node $j+1$ to $j$ is the unique path in $G$ from $j+1$ to $j$, and
  all nodes lie along the path (note that $j=p$ in this case). This
  completes the proof.
\end{proof}

We note the following corollary of Proposition~\ref{prop:directed_sp_trees}:
\begin{corollary} \label{cor:number_spanning_trees}
  For the directed graph $G$ in Fig.~\ref{fig:directed_graph}, the
  number of spanning trees rooted at $j$ is
  \begin{itemize}
  \item $n+2-j$, if $j\in\{1$, \ldots, $n+1\}$ 
  \item $2n+3 -j$, if $j\in\{n+2$, \ldots, $2n+2\}$.
  \end{itemize}
  Consequently the number of spanning trees rooted at $j$ is at most
  $n+1$.
\end{corollary}

  Now we turn to the kernel of $A_\kappa^t$. In
  Corollary~\ref{cor:rk_A-kappa}, we argued that $\ker(A_\kappa^t)$ is
  spanned by a positive vector. This is a consequence of
  \cite[Lemma~2]{TG}, which built on the well-known Matrix-Tree
 Theorem of algebraic combinatorics~\cite[\S 5.6]{Stanley2}, 
 and also gives an explicit formula for this vector. For this, we need some more notation:
  \\
  \noindent
  {\bf Notation}.
  Following \cite{TG}, 
  for a directed spanning 
  tree $T$ of an edge-labeled directed graph $G$, we denote
  by $L(T)$ the product of all edge labels in the spanning
  tree $T$:
  \begin{equation}
    \label{eq:def_L}
    L(T) := \prod_{y_i \overset{a}{\to} y_j \in T} a\ .
  \end{equation}
  Note that $L(T) > 0$, as it is a product of rate constants. 
  %
  \begin{proposition} \label{prop:rho}  
    Recall the spanning trees $T_{j,p}$ of $G$ from
    Fig.~\ref{fig:directed_graph}. For the matrix $\tilde
    A_{\kappa}^t$ displayed in~\eqref{eq:def_Ak_aug} for the
    translated network~\eqref{eq:augmented_network_proc}, the
    nullspace is spanned by the positive vector $\rho \in
    \mathbb{R}_+^{2n+2}$ whose coordinates are given below:
    \begin{equation}
      \label{eq:def_v}
      \rho_j =
      \begin{cases}
        \displaystyle \sum_{p=j}^{n+1} L(T_{j,p}) & {\rm if~} 1\leq j\leq n+1 \\
        \displaystyle \sum_{i=j-(n+1)}^{n+1} L(T_{j,n+1+i}) & {\rm if~} n+2 \leq j
        \leq 2n+2~.
      \end{cases}
    \end{equation}
    The terms $L(T_{j,p})$ are defined in eq.~\eqref{eq:Tjp} below.
  \end{proposition}
  \begin{proof}
    Proposition~\ref{prop:directed_sp_trees} and application of
    \cite[Lemma~2]{TG} to $G$ from Fig.~\ref{fig:directed_graph}.
  \end{proof}

Next we will compute the product $L(T_{j,p})$ associated to each 
spanning tree $T_{j,p}$ of $G$. To this end, we recall the labeling of reactions
between adjacent nodes $j$ and $j+1$ for $1\leq j \leq n-1$:
\begin{displaymath}
  \xymatrix@C=9ex{
    \dots & \overset{j}{\bullet} \ar @<.4ex> @{-^>} ^{k_{2j-1}} [r] &
    \overset{j+1}{\bullet} \ar @{-^>} ^{k_{2j} }[l] & \dots
  }
\end{displaymath}
For a node $j$ with $n+2\leq j \leq 2n+1$, we write $j$ as $j=n+1 +i$
(so, $1\leq i \leq n+1$) and recall the labeling of reactions
between adjacent nodes $j$ and $j+1$: 
\begin{displaymath}
  \xymatrix@C=9ex{
    \dots & \overset{\tiny
      \begin{array}{c}
        (n+1) +(i+1) \\ = \\ j+1
      \end{array}
    }{\bullet} \ar @<.4ex> @{-^>} ^{\ell_{2(n+1-i)}} [r] &
    \overset{\tiny
      \begin{array}{c}
        (n+1) + i \\
        = \\
        j
      \end{array}
    }{\bullet} \ar @{-^>} ^{\ell_{2(n+1-i)}+1}[l] & \dots
  }
\end{displaymath}
Now we use Proposition~\ref{prop:directed_sp_trees} to compute
$L(T_{j,p})$, for a spanning tree $T_{j,p}$ of $G$:
\begin{itemize}
\item if $1\leq j\leq n$, the tree $T_{j,p}$ splits into four paths:
  \begin{enumerate}[{(}a{)}]
  \item $p+1 \to \cdots \to n+2$, with product of edge labels $k_{2n+1}\, \prod_{i=p+1}^n
    k_{2i-1} = \prod_{i=p+1}^{n+1} k_{2i-1}$, 
  \item $n+2\to \cdots \to 1$, with product of edge labels $\ell_1 \prod_{i=1}^n, 
    \ell_{2(n+1-i)+1} = \prod_{i=1}^{n+1} \ell_{2(n+1-i)+1}$,
  \item $1\to \cdots \to j$, with product of edge labels $\prod_{i=1}^{j-1} k_{2i-1}$, 
  \item $p\to \cdots \to j$, with product of edge labels $\prod_{i=j}^{p-1} k_{2i}$.
  \end{enumerate}
\item if $j=n+1$ (so, $p=n+1$,
  by Proposition~\ref{prop:directed_sp_trees}), the tree $T_{j,p}$
  splits into two paths:
  \begin{enumerate}[{(}a{)}]
  \item $n+2\to \cdots \to 1$, with product of edge labels 
    $\prod_{i=1}^{n+1} \ell_{2(n+1-i)+1}$, as in (b) in the previous case.
  \item $1\to \cdots \to n+1$, with product of edge labels $\prod_{i=1}^n k_{2i-1}$.
  \end{enumerate}
\item if $n+2\leq j\leq 2n+1$, write $j = n+1 + j_0$ and $p= n+1+p_0$,
  and then split $T_{j,p}$ into four paths (cf.\
  Fig.~\ref{fig:np_graph}(b)):
  \begin{enumerate}[{(}a{)}]
  \item $p+1 \to \cdots\to 1$, with product of edge labels $\ell_1 \prod_{i=p_0+1}^n
    \ell_{2(n+1-i)+1} = \prod_{i=p_0+1}^{n+1} \ell_{2(n+1-i)+1}$,
  \item $1\to \cdots \to n+2$, with product of edge labels $k_{2n+1} \prod_{i=1}^n k_{2i-1} =
    \prod_{i=1}^{n+1} k_{2i-1}$,
  \item $n+2 \to \cdots \to j$, with product of edge labels $\prod_{i=1}^{j_0-1}
    \ell_{2(n+1-i)+1}$,
  \item $p\to \cdots \to j$, with product of edge labels $\prod_{i=j_0}^{p_0-1} \ell_{2(n+1-i)}$.
  \end{enumerate}
\item if $j=2n+2$ (so, $p=2n+2$, 
  by Proposition~\ref{prop:directed_sp_trees}), the tree $T_{j,p}$
  splits into two paths:
  \begin{enumerate}[{(}a{)}]
  \item $1\to \cdots\to n+2$, with product of edge labels
    $ \prod_{i=1}^{n+1} k_{2i-1}$, as in (b) in the previous case, 
  \item $n+2\to \cdots \to 2n+2$, with product of edge labels
    $\prod_{i=1}^{n}\ell_{2(n+1-i)+1}$.
  \end{enumerate}
\end{itemize}
Thus, by definition~\eqref{eq:def_L}, we obtain for $L(T_{j,p})$,
where for $i_1<i_0$ we adopt the standard convention
$\prod_{i=i_0}^{i_1} \alpha_i := 1$ for the empty product, and, as
before, $j_0 := j-(n+1)$ and $p_0 := p-(n+1)$:
\begin{equation}
  \label{eq:Tjp}
  L(T_{j,p}) =
  \begin{cases}
    \displaystyle \prod_{i=1}^{n+1} \ell_{2(n+1-i)+1} \cdot \prod_{i=1}^{j-1}
    k_{2i-1} \cdot \prod_{i=j}^{p-1} k_{2i} \cdot
    \prod_{i=p+1}^{n+1} k_{2i-1}
    & {\rm if~} 1\leq j\leq n \\ 
    \displaystyle \prod_{i=1}^{n+1} \ell_{2(n+1-i)+1} \cdot \prod_{i=1}^n
    k_{2i-1}
    & {\rm if~} j = n+1 \\
    \displaystyle \prod_{i=1}^{n+1} k_{2i-1} \cdot \prod_{i=1}^{j_0-1}
    \ell_{2(n+1-i)+1} \cdot \prod_{i=j_0}^{p_0-1} \ell_{2(n+1-i)}
    \cdot \prod_{i=p_0+1}^{n+1} \ell_{2(n+1-i)+1}
    & {\rm if~} n+2 \leq j \leq 2n+1 \\
    \displaystyle \prod_{i=1}^{n+1} k_{2i-1} \cdot
    \prod_{i=1}^{n}\ell_{2(n+1-i)+1}      
    & {\rm if~} j = 2n+2 ~. 
  \end{cases}
\end{equation}

\end{document}